\documentclass[a4paper, conference, 10pt]{IEEEtran}
\usepackage[left=1.5cm,right=1.5cm,top=1.8cm,bottom=2.5cm]{geometry}
\IEEEoverridecommandlockouts
\usepackage{amsmath}
\usepackage{array}
\usepackage{cite}
\usepackage{amsmath,amssymb,amsfonts}
\usepackage{algorithmic}
\usepackage{graphicx}
\usepackage{booktabs}  

\usepackage[colorlinks=true,urlcolor=blue,linkcolor=blue,citecolor=blue]{hyperref}
\usepackage{amsthm} 
\usepackage{graphicx}
\usepackage{algorithm}
\usepackage{algorithmic}
\newtheorem{theorem}{Theorem}
\newtheorem{corollary}[theorem]{Corollary}
\newtheorem{proposition}[theorem]{Proposition}
\newtheorem{lemma}[theorem]{Lemma}
\newtheorem{remark}[theorem]{Remark}

\usepackage{graphicx} 
\usepackage{subcaption}
\usepackage{soul}
\usepackage[usenames,dvipsnames]{xcolor}
\usepackage{svg}
\usepackage{cite}
\usepackage{enumitem}
\setlist[enumerate]{topsep=0pt,itemsep=-1ex,partopsep=1ex,parsep=1ex,leftmargin=4ex}
\setlist[itemize]{topsep=0pt,itemsep=-1ex,partopsep=1ex,parsep=1ex,leftmargin=4ex}
\usepackage{caption}

\begin{document}

\title{Optimizing Age of Information in Networks with Large and Small Updates
}

\author{Zhuoyi Zhao, Vishrant Tripathi, and Igor Kadota
\IEEEcompsocitemizethanks{\IEEEcompsocthanksitem Zhuoyi Zhao and Igor Kadota are with the Department of Electrical and Computer Engineering, Northwestern University, USA. E-mail: zhuoyizhao2025@u.northwestern.edu and kadota@northwestern.edu. 

Vishrant Tripathi is with the Department of Electrical and Computer Engineering, Purdue University, USA. E-mail: tripathv@purdue.edu.
}
}

\maketitle

\begin{abstract}
Modern sensing and monitoring applications typically consist of sources transmitting updates of different sizes, ranging from a few bytes (position, temperature, etc.) to multiple megabytes (images, video frames, LIDAR point scans, etc.). Existing approaches to wireless scheduling for information freshness typically ignore this mix of large and small updates, leading to suboptimal performance. In this paper, we consider a single-hop wireless broadcast network with sources transmitting updates of different sizes to a base station over unreliable links. Some sources send large updates spanning many time slots while others send small updates spanning only a few time slots. Due to medium access constraints, only one source can transmit to the base station at any given time, thus requiring careful design of scheduling policies that takes the sizes of updates into account. First, we derive a lower bound on the achievable Age of Information (AoI) by any transmission scheduling policy. Second, we develop optimal randomized policies that consider both switching and no-switching during the transmission of large updates. Third, we introduce a novel Lyapunov function and associated analysis to propose an AoI-based Max-Weight policy that has provable constant factor optimality guarantees. Finally, we evaluate and compare the performance of our proposed scheduling policies through simulations, which show that our Max-Weight policy achieves near-optimal AoI performance.
\end{abstract}


\section{Introduction}

The Age of Information (AoI) metric has received significant attention in the literature~\cite{kaul2011minimizing,AoI_V,AoIUAV1,9637803,WiSwarm,yu2022age,beytur2020towards,abd2019role} due to its relevance for emerging time-sensitive applications such as connected autonomous vehicles~\cite{kaul2011minimizing,AoI_V}, cooperative UAV swarms~\cite{AoIUAV1,9637803,WiSwarm}, and the Internet-of-Things~\cite{yu2022age,beytur2020towards,abd2019role}. AoI captures the freshness of information from the destination's perspective by measuring the time elapsed since the generation of the most recent update. In many such applications, the content being transmitted is multimodal, including a mix of small information updates (such as position, temperature, pressure, etc.) and large updates (such as video frames, LIDAR point scans, images, etc.). Given the slotted nature of modern communication networks (e.g., OFDM in WiFi~6 and in 5G), the transmission of a single update may require multiple time slots, where each time slot carries an individual data packet. 


\begin{figure}
    \centering
    \includegraphics[width=0.5\linewidth]{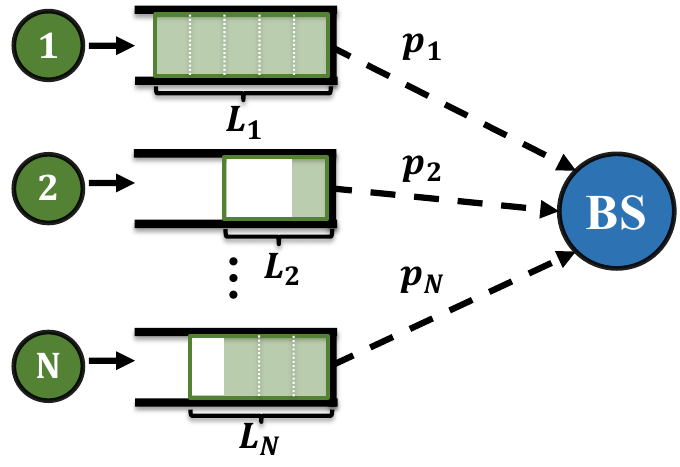}
    \vspace{-0.5em} 
    \caption{Network with $N$ sources transmitting information updates to a base station (BS). Sources generate updates over time and keep only the freshest update. Updates from source $i\in\{1,\ldots,N\}$ are composed of $L_i$ data packets. The BS selects one source at every time slot~$t$ to transmit a single packet via its unreliable wireless link.}
    \label{fig: system model}
    \vspace{-1em} 
\end{figure}

In this paper, we consider a network with multiple sources transmitting time-sensitive updates to a base station (BS), as illustrated in Fig.~\ref{fig: system model}. 
We assume that information updates generated by source $i\in\{1,2,\ldots,N\}$ are composed of $L_i$ data packets.  
Further, we assume that, in each time slot, the BS can schedule one source to transmit a single data packet, and that these transmissions are unreliable. 
Our goal is to develop transmission scheduling policies that attempt to optimize information freshness in the network. 

\begin{figure}[t]
\centering
\vspace{-0.5em} 
\begin{minipage}{\linewidth}
    \centering
    \includegraphics[width=0.9\linewidth]{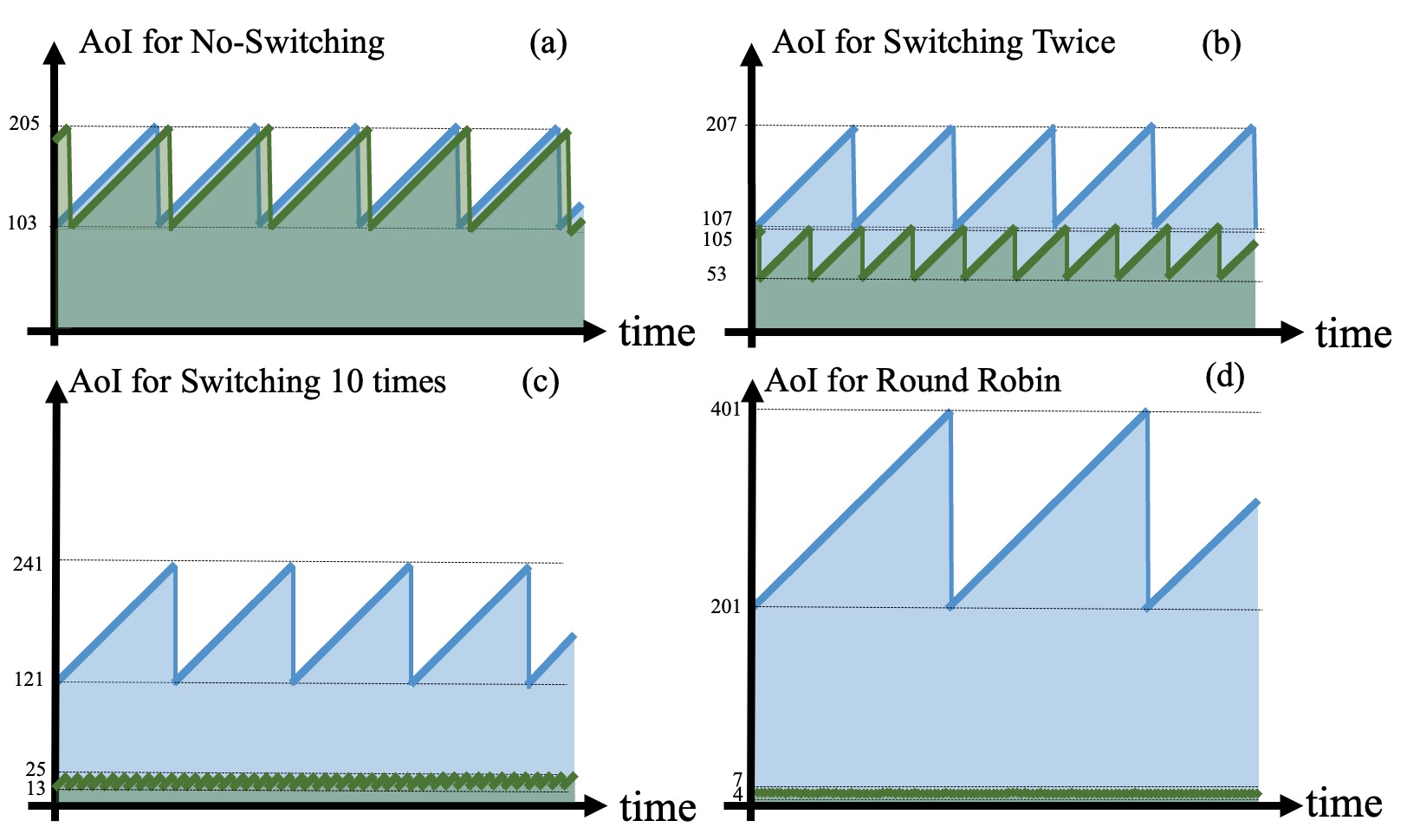}
    \caption{AoI evolution in a two-source network with reliable channels and sources 1 and 2 with update lengths $L_1=100$ and $L_2=2$, respectively. Each of the four plots are associated with a different scheduling policy: (a) no-switching; (b) switching twice during the update of source~1; (c) switching 10 times during the update of source~1; and (d) round robin. The average AoI is shown in Table~\ref{table:2sources}}
    \label{fig:simple_non_sysmmetric}
\end{minipage}
\vspace{-1.5em} 
\end{figure}

\begin{figure}[t]
\centering
\begin{minipage}{1\linewidth}
    \centering
    \small 
    \captionof{table}{Average AoI associated with the plots in Fig.~\ref{fig:simple_non_sysmmetric}.}
    \label{table:2sources}
    \begin{tabular}{cccc}
    \toprule
    \textbf{Scheduling Policy} & \textbf{Average AoI} \\
    \midrule
    (a) No-Switching        & 154 \\
    (b) Switching Twice     & 118 \\
    (c) Switching 10 times  & \textbf{99} \\
    (d) Round Robin         & 153.25 \\
    \bottomrule
    \end{tabular}
\end{minipage}
\vspace{-1.5em}
\end{figure}

\emph{Developing effective scheduling policies in networks with different update sizes is challenging}. To illustrate this challenge, we consider a simple two-source network with reliable channels in which source 1 transmits large updates, each with \(L_1=100\) packets, and source 2 transmits small updates, each with \(L_2=2\) packets. In each time slot, only one source can be scheduled by the BS. 
In Fig.~\ref{fig:simple_non_sysmmetric}, we compare the AoI evolution associated with four scheduling policies: 
(a) no-switching policy, which delivers \emph{complete updates} from sources 1 and 2 in turns; 
(b) switching twice policy, which periodically delivers two small updates from source 2 for every large update from source 1; 
(c) switching 10 times policy, which periodically delivers ten small updates for every large update; 
and (d) the round robin policy, which schedules \emph{packet transmissions} from sources 1 and 2 in turns. 
Table~\ref{table:2sources} reports the average AoI achieved by the different policies, indicating that, even in small networks with reliable channels, judiciously accounting for the different update lengths can significantly improve AoI. 

However, most scheduling policies proposed in prior works~\cite{kadota2019minimizing,zakeri2023minimizing,zhou2019joint,tang2020minimizing ,kadota2019scheduling,fountoulakis2023scheduling ,atay2021aging,liu2024optimizing,zhao2025optimizing} are designed under the assumption that every successful transmission of a fresher packet to the destination leads to an AoI reduction, which is only true if $L_i=1,\forall i$, and can lead to poor AoI performance. Table~\ref{table:2sources} shows that the two policies that disregard the difference in update lengths, namely No-Switching and Round Robin, have an average AoI that is at least $54\%$ worse than Switching 10 times. 

\noindent\textbf{Main Contributions.} In this paper, we address the problem of AoI optimization in wireless networks in which sources may have different update lengths. Our main contributions can be summarized as follows:
\begin{itemize}[leftmargin=0.15in]
\item We find the \emph{optimal schedulers} for two classes of policies: (i) Switching Randomized Policies (SRP), in which the BS randomly selects a source for transmission in every slot~$t$; and (ii) No-Switching Randomized Policies (NSRP), in which once the first packet of an update is successfully transmitted to the BS, the BS must continuously select the same source until the entire information update is delivered. 
We obtain \emph{closed-form analytical expressions} for the AoI performance of SRP and NSRP. Using our universal lower bound, which generalized prior works~\cite{kadota2018scheduling,li2019minimizing}, we derive a constant factor optimality guarantee for the optimal SRP. 
\item We develop a novel low-complexity Max-Weight policy that makes scheduling decisions based on: AoI, system time, and \emph{remaining number of packets in the current update}. We derive a \emph{constant factor optimality guarantee} for the Max-Weight policy. 
\emph{To the best of our knowledge, this is the first policy with a constant factor optimality guarantee in terms of AoI for networks with different update lengths and unreliable channels.} 
\item To derive performance guarantees, we propose a \emph{novel Lyapunov function and analysis}. Traditional AoI-based Lyapunov functions and analysis are insufficient since there can be long periods of time when the AoI does not change despite successful packet deliveries (due to large update sizes and unreliable channels). We add notions of system time, waiting time, and throughput debt to our Lyapunov function and utilize these in Lyapunov drift analysis to obtain our performance bounds. 
\item We evaluate the impact of the network configuration and scheduling policy on AoI. Our numerical results show that, the performance of the Max-Weight policy is near optimal for a wide variety of network settings.
\end{itemize}

\noindent\textbf{Related Work.}
The design of transmission scheduling policies that optimize the AoI in wireless networks has been extensively investigated (e.g.,~\cite{kadota2019minimizing,zakeri2023minimizing,zhou2019joint,tang2020minimizing ,kadota2019scheduling,fountoulakis2023scheduling ,atay2021aging,liu2024optimizing,zhao2025optimizing,li2019minimizing,tripathi2021computation,zhou2019minimum,AoI5G}). Various network configurations have been considered, including those with stochastic arrivals~\cite{kadota2019minimizing,zakeri2023minimizing}, energy constraints~\cite{zhou2019joint,tang2020minimizing}, throughput constraints~\cite{kadota2019scheduling,fountoulakis2023scheduling }, and imperfect knowledge~\cite{atay2021aging,liu2024optimizing,zhao2025optimizing}. 
Most prior works~\cite{kadota2019minimizing,zakeri2023minimizing,zhou2019joint,tang2020minimizing ,kadota2019scheduling,fountoulakis2023scheduling ,atay2021aging,liu2024optimizing,zhao2025optimizing} assume that each update consists of a single packet, while a few recent studies~\cite{abd2019role,AoI5G,li2019minimizing,tripathi2021computation,zhou2019minimum} considered networks with different update sizes. Most related to this paper are~\cite{li2019minimizing,tripathi2021computation,zhou2019minimum}. 

In~\cite{li2019minimizing}, Li \emph{et al.} consider networks with reliable channels and sources that generate updates with different sizes. 
They develop the \textit{Juventas} scheduling policy based on the ``AoI outage'' defined as the difference between AoI and system time and provide performance guarantees under the assumption that each update can be fully delivered within one time slot. The interesting insights in~\cite{li2019minimizing} are limited to networks with reliable channels. 
Similarly, in~\cite{tripathi2021computation}, Tripathi \emph{et al.} propose a low-complexity Whittle index resource allocation algorithm for networks with reliable channels and non-uniform update lengths. 
This algorithm assumes that complete updates must be transmitted before switching sources, and it lacks performance guarantees in terms of AoI. 
In~\cite{zhou2019minimum}, Zhou \emph{et al.} study AoI minimization by jointly designing sampling and scheduling policies. They derive the Bellman equation, unveil interesting structural properties of the solution, apply linear decomposition method to decouple sources, and develop a structure-aware algorithm that solve the Bellman equation for each source in parallel to compute a sub-optimal policy. The proposed structure-aware algorithm has no performance guarantees and it 
has a computational complexity of \(\mathcal{O}(\bar{h}^2L)\), where \(\bar{h}\) is the imposed upper bound on the AoI and \(L\) is the update length, which may limit its practical applicability. 

In contrast, \emph{in this paper we consider wireless networks where sources generate updates of different lengths and the wireless channels are unreliable}. We develop dynamic scheduling policies with constant factor optimality guarantees in terms of AoI. Further, our proposed scheduling schemes are low complexity - their complexity scales linearly with the number of sources \textit{and does not scale at all with the size of the updates $L$.}

The remainder of this paper is organized as follows. In Sec.~\ref{sec: network model}, we describe the network model. In Sec.~\ref{sec:lower_bound}, we derive a lower bound on the achievable AoI. In Sec.~\ref{sec: Randomized Policies}, we develop and analyze the optimal SRP and optimal NSRP. We also use the lower bound derived from earlier to prove performance guarantees for the optimal SRP and NSRP. In Sec.~\ref{sec: Max Weight}, we develop and analyze the Max-Weight policy, and provide performance bounds. In Sec.~\ref{sec:simulation}, we provide detailed numerical results that illustrate the performance gains of our approach in a wide variety of network settings.

\section{Network Model}\label{sec: network model}

Consider a single-hop wireless network with a base station (BS) that receives time-sensitive updates from $N$ sources, as illustrated in Fig.~\ref{fig: system model}. Let time be slotted, with slot index $t \in \{1,2,\ldots,T\}$, where $T$ denotes the time horizon. 
Each \emph{information update} generated by source $i \in \{1, 2, \ldots, N\}$ consists of $L_i$ \emph{data packets}, where each \emph{packet} can be entirely transmitted in one time slot. An \emph{information update} is deemed to have been delivered successfully only after all $L_i$ \emph{data packets} reach the BS.
However, due to interference and capacity constraints, only one source can transmit in any given time slot, and the BS can receive at most one packet per slot, which may not constitute an entire update. These limitations necessitate the careful design of scheduling policies that account for AoI, the size of updates, and the number of packets remaining in queues. 
Next, we discuss the \emph{update} generation process and \emph{packet} transmission process in our system model.

\noindent\textbf{Update Generation Process.} 
Updates generated by each source $i$ are queued in a corresponding single update buffer queue. At any time slot, the queue contains all the packets that are remaining for transmission from the latest generated information update. 
Each source decides whether to generate a new update and place it in its buffer by looking at the number of packets remaining in the queue. Specifically, if the buffer is empty, then source $i$ generates a new update and places it in the buffer. Similarly, if the buffer is full, and the source is not currently transmitting, then it generates a new update and places it in the buffer. This generation policy ensures that the source always transmits the freshest available update, when it starts transmission. We assume that the update generation and transmission is non-preemptive, i.e. if a part of the update remains undelivered, then the source keeps the remainder of the current update in the buffer and does not replace it with a new update. 
Intuitively, this queuing discipline helps reduce the age of information by avoiding partial transmissions that do not reduce AoI, while ensuring that newly transmitted updates remain as fresh as possible.

\noindent\textbf{Packet Transmission Process.} 
In each slot \(t\), the BS either idles or selects one source for transmission. 
Let \(u_i(t) = 1\) indicate that source \(i\) is selected during slot \(t\), and \(u_i(t) = 0\) otherwise. It follows that
$\sum_{i=1}^N u_i(t) \leq 1, \forall t.$
The selected source attempts to transmit \textit{one} packet from its queue to the BS over an unreliable wireless channel. Let \(c_i(t) = 1\) indicate that the channel from source \(i\) to the BS is ON during slot \(t\), and \(c_i(t) = 0\) indicates otherwise. The channel states \(c_i(t)\) are i.i.d.\ over time and independent across different sources, with \(\mathbb{P}(c_i(t)=1) = p_i \in (0, 1]\) for all \(i, t\). 

Let \(d_i(t)\) be an indicator such that \(d_i(t)=1\) if source \(i\) successfully transmits a packet in slot \(t\), and 0 otherwise. A transmission is successful if the source is scheduled and the channel is ON, implying 
$d_i(t) = c_i(t)\,u_i(t), \forall i,t$.
Since the BS does not know the channel states before making scheduling decisions, \(u_i(t)\) and \(c_i(t)\) are independent, which yields
$\mathbb{E}[d_i(t)] = p_i\,\mathbb{E}[u_i(t)], 
 \forall i,t$.

Without loss of generality, we assume that \emph{at the beginning of each slot $t$, the update generation occurs before packet transmission can start.} Next, we introduce network performance metrics of interest and then formulate the AoI minimization problem. 

\noindent\textbf{Remaining Update Length.} 
Let $L_i(t) \in \{1, 2, \ldots,$ $ L_i\}$ denote the number of packets remaining to be transmitted in source $i$'s queue at the beginning of slot $t$, after the update generation process. 
The evolution of $L_i(t)$ is given by
\begin{equation}\label{Levolves}
L_{i}(t+1) = 
\begin{cases}
L_i, & \text{if } d_i(t) = 1 \text{ and } L_i(t) = 1,\\
L_{i}(t) - 1, & \text{if } d_i(t) = 1 \text{ and } L_i(t) \neq 1,\\
L_i(t), & \text{if } d_i(t) = 0.
\end{cases}
\end{equation}
The remaining update length $L_i(t)$ is critical for AoI tracking, as it determines the number of packet deliveries required before the AoI can be reduced.

\noindent\textbf{System Time.} The \emph{system time} of the update in the queue of source $i$ in slot \(t\) is defined as \(z_i(t) := t - \tau^S_i(t)\), where \(\tau^S_i(t)\) represents the time at which the update was generated (i.e., the ``source timestamp''). 
The system time evolves as
\begin{equation}\label{zevolves}
z_{i}(t+1) = 
\begin{cases}
1, & \text{if } d_i(t) = 0 \text{ and } L_i(t) = L, \\
  & \text{or } d_i(t) = 1 \text{ and } L_i(t) = 1,\\
z_{i}(t) + 1, & \text{otherwise}.
\end{cases}
\end{equation}
System time $z_{i}(t)$ is crucial for tracking AoI, as it measures how fresh the update is before delivery.

\noindent\textbf{Age of Information.} 
Let \( h_i(t) := t - \tau^D_i(t) \) be the AoI associated with source \( i \) at the beginning of slot \( t \), where $\tau^D_i(t) $ is the generation time of the last delivered update. 
The evolution of \( h_i(t) \) is given by:
\begin{equation}\label{hevolves}
h_i(t+1) = \begin{cases} 
z_i(t) + 1 & \text{if } d_i(t) = 1 \text{ and } L_i(t)=1, \\ 
h_i(t) + 1 & \text{otherwise.} 
\end{cases}
\end{equation}
We assume that \( h_i(1) = 1 \), \( z_i(1) = 0 ,\forall i \), and \( L_i(1) = L_i ,\forall i \).

\noindent\textbf{Long-term Packet Throughput.}
The long-term packet throughput of source \(i\) is given by
\begin{equation}\label{throughput define}
    q_i = \lim_{T\rightarrow \infty }\frac{\mathbb{E}[D_i(T)]}{T}, 
\end{equation}
where $D_i(T ) = \sum^T_{t =1} d_i(t)$ is the total number of information updates delivered from source $i$ by the end of the time-horizon $T$. 
The shared and unreliable wireless channel restricts the set of feasible values of long-term throughput. By employing $\mathbb{E}[d_i(t)] = p_i\mathbb{E}[u_i(t)]$ and $\sum^N_{i=1} u_i(t)\leq 1$ into the definition of long-term throughput in~\eqref{throughput define}, we obtain
\begin{equation}\label{throughput constraint}
   \frac{\mathbb{E}[D_i(T)]}{T} =\frac{p_i\sum_{t=1}^T\mathbb{E}[u_i(t)]}{T}= \sum_{i=1}^N \frac{q_i}{p_i} \leq 1.
\end{equation}

\noindent\textbf{AoI minimization problem.} The transmission scheduling policies considered in this paper are non-anticipative, which means that they do not use future information when making scheduling decisions. Let $\Pi$ represent the class of non-anticipative policies and let $\pi \in \Pi$ denote an arbitrary admissible policy. 
To capture the information freshness in a network employing  policy \( \pi \in \Pi \), we define the Expected Weighted Sum AoI (EWSAoI) in the limit as the time horizon $T$ grows to infinity as
\begin{equation}\label{EWSAoIexpression}
    \mathbb{E}[J^\pi] = \lim_{T \to \infty} \frac{1}{TN} \sum_{t=1}^T \sum_{i=1}^N \alpha_i \mathbb{E}[h_i^{\pi}(t)], 
\end{equation}
where \( \alpha_i >0 \) represents the priority of source \( i \). We denote by $\pi^*\in \Pi$ the AoI-optimal policy that achieves minimum EWSAoI, namely
\begin{center}
\fbox{
  \parbox{0.8\linewidth}{
  \centering
  \vspace{-0.5em} 
    \begin{align}
    \mathbb{E}\left[J^*\right] = \min_{\pi \in \Pi}\lim_{T \to \infty}& \frac{1}{TN} \sum_{t=1}^T \sum_{i=1}^N \alpha_i \mathbb{E}[h_i^{\pi}(t)], \\
    \text{s.t.}&\sum_{i=1}^N u_i(t) \leq 1.
    \end{align}
  \vspace{-0.5em} }
}
\end{center}
where $\mathbb{E}\left[J^*\right]$ is the EWSAoI associate with policy $\pi^*$, and the expectation is with respect to the randomness in the channel state $c_i (t )$ and in scheduling decisions $u_i (t )$. Next, we derive a universal lower bound for the AoI minimization problem.

\section{Lower Bound}\label{sec:lower_bound}
In this section, we derive a lower bound on the achievable EWSAoI under any admissible scheduling policy \(\pi \in \Pi\). 
We first define \emph{waiting time} and \emph{service time}, then we characterize the EWSAoI in terms of these two quantities, and, finally, we derive the lower bound.

\noindent\textbf{Waiting Time and Service Time.} Let \(K_i(T)\) denote the total number of delivered updates from source $i$ by the end of slot \(T\) and let $m\in\{1,...,K_i(T)\}$ be the index of the delivered updates from source $i$. 
Let \(t'_i[m]\) denote the time slot in which the first packet of the \(m\)th delivered update is received, and let \(t_i[m]\) denote the time slot in which the last packet of the \(m\)th delivered update is received. 
We define the \emph{waiting time} of the \(m\)th update from source \(i\) as $W_i[m] \;:=\; t'_i[m] - t_i[m-1]$, which is the interval between the delivery of the last packet of the \((m-1)\)th update and the delivery of the first packet of the \(m\)th update. Similarly, the \emph{service time} of the \(m\)th update for source \(i\) is defined as $S_i[m] \;:=\; t_i[m] - t'_i[m]$, which is the interval between the delivery of the first and last packets of the \(m\)th update. We assume that \(t_i[0] = 0\), \(W_i[0] = 0\), and \(S_i[0] = 0\) for all \(i\).

For a set of values \(\mathbf{x}\), let \(\Bar{\mathbb{M}}[\mathbf{x}]\) denote the sample mean. The time horizon \(T\) is omitted in the notation \(\Bar{\mathbb{M}}[\cdot]\) for simplicity. Using this operator, the sample mean of $W_i[m]$ and $S_i[m]$ for a fixed source $i$ is given by
\begin{equation}
    \Bar{\mathbb{M}} \bigl[W_i[m]\bigl] \;:=\; \frac{1}{K_i(T)} \sum_{m=1}^{K_i(T)} W_i[m], 
\end{equation}
\begin{equation}
\Bar{\mathbb{M}} \bigl[S_i[m]\bigl] \;:=\; \frac{1}{K_i(T)} \sum_{m=1}^{K_i(T)} S_i[m].
\end{equation}

\begin{proposition}\label{prop: AoI expression with waiting/service time}
The infinite-horizon Weighted Sum AoI achieved by scheduling policy $\pi$, i.e. \(J^\pi\), can be written as
\begin{equation}\label{performance for waiting and service time}
\begin{aligned}
J^\pi 
&= \lim_{T\rightarrow \infty }\sum_{i=1}^{N}\frac{\alpha_i}{N}
\Biggl(
  \frac{\Bar{\mathbb{M}}\bigl[(W_i[m]+S_i[m])^2\bigr]}
       {2\,\Bar{\mathbb{M}}\bigl[W_i[m]+S_i[m]\bigr]} \\
&\quad +\;
  \frac{\Bar{\mathbb{M}}\bigl[S_{i}[m-1]\,(W_i[m]+S_i[m])\bigr]}
       {\Bar{\mathbb{M}}\bigl[W_i[m]+S_i[m]\bigr]}
  \;+\; 1
\Biggr),
\end{aligned}
\end{equation}
where \(W_i[m]\) and \(S_i[m]\) are the waiting time and service time of the \(m\)th update from source \(i\).
\end{proposition}
\begin{proof}
Using a sample path argument, we compute the sum of AoI during each update, and take the average over time by rewriting the time horizon \(T\) as the sum of waiting and service times. By omitting zero-order terms, we obtain~\eqref{performance for waiting and service time}. Detailed derivations are provided in Appendix A.
\end{proof}

\begin{remark}\label{rem:service_time}
Equation~\eqref{performance for waiting and service time} holds for any scheduling policy \(\pi \in \Pi\) and generalizes known results for the single-packet case~\cite{kadota2019scheduling} to the scenario where each update may contain multiple packets. The first term on the RHS of~\eqref{performance for waiting and service time}
depends on both the waiting time and service time. The second term depends on the previous update's service time and the sum of the current update's waiting time and service time. Intuitively, to minimize AoI, the scheduling policy should attempt to deliver packets from a source that currently has high waiting time and high service time, especially the latter.
\end{remark}

Based on Proposition~\ref{prop: AoI expression with waiting/service time}, we now establish a universal lower bound on the achievable AoI.

\begin{theorem}\label{theo:lower_bound}
For a network with parameters \(\{N,\alpha_i,p_i,L_i\}\), the following bound holds for all admissible policies \(\pi \in \Pi\):
\begin{equation}\label{performancelowerbound}
L_{B} =\frac{1}{2N}\Biggl(\sum_{i=1}^{N}\sqrt{\frac{\alpha_i\,L_{i}}{p_{i}}}\Biggr)^{2}+\sum_{i=1}^{N}\alpha_i
\leq\mathbb{E}[J^\pi],
\end{equation}
\end{theorem}
\begin{proof} Applying Jensen Inequality and $S_{i}[m-1](W_i[m]+S_i[m])>0, \forall i,m,$, we obtain a lower bound on~\eqref{performance for waiting and service time}. Then, by using Cauchy-Schwarz Inequality to solve the optimization problem with respect to throughput of each sources yields~\eqref{performancelowerbound}. 
Detailed derivations are provided in Appendix B.
\end{proof}

\vspace{-1em}
\section{Randomized Policies}\label{sec: Randomized Policies}


In this section, we discuss two classes of randomized scheduling policies:
Switching Randomized Policies, in which the BS randomly selects a source for transmission in every slot~$t$; and
No-Switching Randomized Policies, in which once the first packet of an update is successfully transmitted to the BS, the BS must continuously select the same source until the entire information update is delivered. In Sec.~\ref{sec:compare}, we compare the performance of these two classes of randomized policies in symmetric and non-symmetric networks. 

\subsection{Switching Randomized Policies (SRP)}\label{sec: Switching Randomized Policy}
Let \(\Pi^s_r\) denote the class of SRPs. A BS running a policy \(s \in \Pi^s_r\) operates as follows: in each slot~$t$, the BS selects source~\(i\) with scheduling probability \(\mu^s_i \in (0,1]\), where the probabilities satisfy $\sum_{i=1}^N \mu^s_i \leq1.$
If source $i$ is selected during slot $t$, then it transmits a packet to the BS. SRPs select sources at random, without taking into account the current AoI $h_i(t)$, system time $z_i(t)$, nor the number of remaining packet $L_i(t)$ at each source.  Each policy $s$ is fully characterized by the set of scheduling probabilities $\{\mu^s_i \}^N_{i=1}$. Note that under an SRP, packets from different sources are interleaved between one another, it is not necessary for all packets belonging to an update to be delivered continuously.

Next, we obtain the optimal SRP and provide performance guarantees for it in terms of AoI. 
Specifically, Proposition~\ref{prop:EWSAoIRandomzedproposition} provides the EWSAoI associated with an arbitrary SRP $s \in \Pi_s$ and Theorem~\ref{theo:stationary_performance} characterizes the optimal SRP $S$ and its performance guarantee.

\begin{proposition}\label{prop:EWSAoIRandomzedproposition}
For a network with parameters \(\{N,\alpha_i,p_i,L_i\}\) and any SRP \(s \in \Pi^s_r\) characterized by \(\{\mu^s_i\}_{i=1}^N\), the corresponding EWSAoI is given by
\begin{equation}\label{eq:EJ^s}
\mathbb{E}[J^s] = \frac{1}{N} \sum_{i=1}^N \alpha_i \left( \frac{3 L_i - 1}{2 p_i \mu^s_i } + 1 \right).
\end{equation}
\end{proposition}
\begin{proof}
First, we take the expectation of~\eqref{performance for waiting and service time} to obtain an expression for the EWSAoI. Then, we substitute the first and second moments of the waiting time and service times -- which follow a negative binomial distribution~\cite{ross2014introduction} -- to obtain~\eqref{eq:EJ^s}. 
Detailed derivations are provided in Appendix C.
\end{proof}

From~\eqref{eq:EJ^s}, we can obtain the optimal SRP $S$ by solving
\begin{align}\label{OPT:EJ^s}
\min_{\{\mu^s_i\}_{i=1}^N}\sum_{i=1}^N  \alpha_i\biggl(\frac{3\,L_i -1}{2\,p_i\,\mu^s_i}\biggr)
\quad \text{s.t.}\quad 
\sum_{i=1}^N \mu^s_i \;\le\; 1.
\end{align}

\begin{theorem}\label{theo:stationary_performance}
For a network with parameters \(\{N,\alpha_i,p_i,L_i\}\), let \(S\in \Pi^s_r\) be the optimal SRP. Its scheduling probabilities \(\{\mu^S_i\}_{i=1}^N\) are given by
\begin{equation}\label{eq:solutionRandomized}
    \mu^S_i 
    \;=\; 
    \frac{\sqrt{\frac{\alpha_i\,(3\,L_i -1)}{2\,p_i}}}
         {\sum_{j=1}^N \sqrt{\frac{\alpha_j\,(3\,L_j -1)}{2\,p_j}}},
    \quad 
    \forall i.
\end{equation}
The associated EWSAoI is given by 
\begin{equation}\label{Eq:J^S}
    \mathbb{E}[J^S] 
    \;=\; 
    \frac{1}{N}  \sum_{i=1}^N \alpha_i
    \;+\;
    \frac{1}{N} 
    \left(\sum_{i=1}^N \sqrt{\tfrac{\alpha_i\,(3\,L_i-1)}{2\,p_i}}\right)^2.
\end{equation}
which satisfies
\begin{equation}\label{eq:Randomized_guarantee}
    L_B \;\le\; \mathbb{E}[J^S] \;\le\; \rho^S\,L_B,
\end{equation}
where \(L_B\) is the lower bound from Theorem~\ref{theo:lower_bound} and the optimality ratio \(\rho^S\) is
\begin{equation}\label{eq:SRP_opt_ratio}
    \rho^S=\frac{\frac{1}{N}\Bigl(\sum_{i=1}^N \sqrt{\tfrac{\alpha_i(3L_i-1)}{2p_i}}\Bigr)^2+\frac{1}{N}\sum_{i=1}^N \alpha_i}{\frac{1}{2N}(\sum_{i=1}^{N}\sqrt{\frac{\alpha_iL_{i}}{p_{i}}})^{2}+\frac{1}{N}\sum_{i=1}^{N}\alpha_i}
\end{equation}
\end{theorem}
\begin{proof}
To solve~\eqref{OPT:EJ^s} we apply Cauchy-Schwarz Inequality and immediately obtain the optimal SRP in~\eqref{eq:solutionRandomized}. Substituting~\eqref{eq:solutionRandomized} into~\eqref{eq:EJ^s} yields the EWSAoI expression in~\eqref{Eq:J^S}. Comparing~\eqref{Eq:J^S} and~\eqref{performancelowerbound} gives~\eqref{eq:Randomized_guarantee} and~\eqref{eq:SRP_opt_ratio}. 
Detailed derivations are provided in Appendix D.
\end{proof}

Notice that when \(L_i=1,\forall i\), the optimal SRP coincides with that in~\cite{kadota2018scheduling} and achieves an optimality ratio of \(\rho^S=2\). In the more realistic and general case when update lengths \(\{L_i\}_{i=1}^N\) are arbitrary, the optimal SRP attains an optimality ratio in the range \(\rho^S\in[2,3)\).

\subsection{No-Switching Randomized Policies (NSRP)}\label{sec: No-Switching Randomized Policy}


We denote by \(\Pi^{ns}_r\) the class of NSRPs. In contrast to SRPs, NSRPs  \(ns \in \Pi^{ns}_r\) do not switch sources between updates. 
Specifically, once the first packet of an update is successfully transmitted (i.e., \(d_i(t)=1\) and \(L_i(t)=L_i\)), the BS continues selecting the same source for transmission until the update 
delivery is complete. In the slot following the successful transmission of the last packet of an update, the BS selects any source~\(i\) with scheduling probability \(\mu^{ns}_i \in (0,1]\), with $\sum_{i=1}^N \mu^{ns}_i \leq1$. Random selection of sources, according to $\{\mu^{ns}_i\}_{i=1}^N$, continues until the first packet is successfully transmitted. 
%
%
Intuitively, NSRPs reduce the service time $S_i[m]$ by continuously transmitting an entire update from a single source, which is in line with the discussion in Remark~\ref{rem:service_time}. 


NSRPs do not consider the current AoI \(h_i(t)\) nor the system time \(z_i(t)\). However, NSRPs behave differently in case $L_i(t)=L_i$, when scheduling decisions are randomized $\{\mu^{ns}_i\}_{i=1}^N$, and in case $L_i(t)<L_i$, when scheduling decisions are deterministic.   
Each policy $ns$ is fully characterized by the set of scheduling probabilities \(\{\mu^{ns}_i\}_{i=1}^N\). 
Proposition~\ref{prop:EWSAoIRandomzedpropositionNS} provides an expression for the EWSAoI associated with an arbitrary NSRP $ns \in \Pi^{ns}_r$.

\begin{proposition}\label{prop:EWSAoIRandomzedpropositionNS}
For a network with parameters \(\{N,\alpha_i,p_i,L_i\}\) 
and an arbitrary NSRP \(ns \in \Pi^{ns}_r\) with scheduling 
probabilities \(\{\mu^{ns}_i\}_{i=1}^N\),
the EWSAoI is given by:
\begin{equation}\label{eq:EJ^{ns}}
\begin{aligned}
  \mathbb{E}[J^{ns}] 
  &=  \sum_{i=1}^N \frac{\alpha_i}{N} \Biggl[\frac{\mu^{ns}_i p_i}{\sum^N_{j=1}\mu^{ns}_j L_j}\Biggl(\frac{\mathbb{E}[W^2_i] + \mathbb{E}[S^2_i]}{2} 
   \\
    &+ \left(\frac{L_i-1}{p_i}\right)^2
+ 2\mathbb{E}[W^2_i]
      \left(\frac{L_i-1}{p_i}\right)\Biggl) +1
  \Biggr]. 
\end{aligned}
\end{equation}
Here, \(\mathbb{E}[S_i^2]\), \(\mathbb{E}[W_i]\), and \(\mathbb{E}[W_i^2]\) denote the second moment of the service time, the first moment of the waiting time, and the second moment of the waiting time, respectively. The term \(\mathbb{E}[Y_i^2]\) represents the second moment of the number of time slots between two consecutive transmissions from source $i$. 
These quantities are given by
\begin{align}
    \mathbb{E}[S^2_i] &= \frac{(L_i - 1)(L_i - p_i)}{p_i^2}, \\
    \mathbb{E}[W_i]&=\frac{\sum^N_{j=1}\mu^{ns}_j L_j - \mu^{ns}_i(L_i-1)}{\mu^{ns}_i p_i},\\
    \mathbb{E}[W_i^2] &= \frac{1}{1- \sum_{j}\mu^{ns}_j +p_i\mu^{ns}_i }\Biggl[\mu^{ns}_i\left(1 + 2(1 - p_i)\mathbb{E}[W_i]\right) \nonumber\\&+ \sum_{j \neq i}\mu^{ns}_j\left(\mathbb{E}[Y_j^2] + 2L_j\mathbb{E}[W_i]\right)\Biggl], \\
   \mathbb{E}[Y_i^2]&= 2L_i -1 + \frac{(L_i -1)(L_i - p_i)}{p_i}.
\end{align}

\end{proposition}
\begin{proof}
First, we take the expectation of~\eqref{performance for waiting and service time} to obtain an expression for the EWSAoI. For a network employing a NSRP, the service time follows a negative binomial distribution~\cite{ross2014introduction} and the waiting time can be modeled as a Markov chain, from which its first and second order moments are derived via recurrence time analysis. Substituting these results into the EWSAoI expression yields~\eqref{eq:EJ^{ns}}. Detailed derivations are provided in Appendix E.
\end{proof}

From the expression for the EWSAoI in~\eqref{eq:EJ^{ns}}, we can find the optimal scheduling probabilities $\{\mu^{NS}_i\}_{i=1}^N$ by solving the optimization problem below: 
\begin{align}\label{OPT:EJ^ns}
\min_{\{\mu^{ns}_i\}_{i=1}^N}\mathbb{E}[J^{ns}], \quad   \text{s.t.}  \sum_{i=1}^N \mu^{ns}_i\leq 1
\end{align}
The complex expression for the EWSAoI in~\eqref{eq:EJ^{ns}} does not lend itself for a closed-form solution for the optimal scheduling probabilities $\{\mu^{NS}_i\}_{i=1}^N$. However,~\eqref{OPT:EJ^ns} is a convex optimization problem that can be solved numerically. In Sec.~\ref{sec:simulation}, we use a numerical solver to obtain the values of $\{\mu^{NS}_i\}_{i=1}^N$.

\subsection{Comparison of Randomized Policies}\label{sec:compare}

\begin{figure}
    \centering
   \includegraphics[trim=0 30 0 30,clip,width=0.9\linewidth]{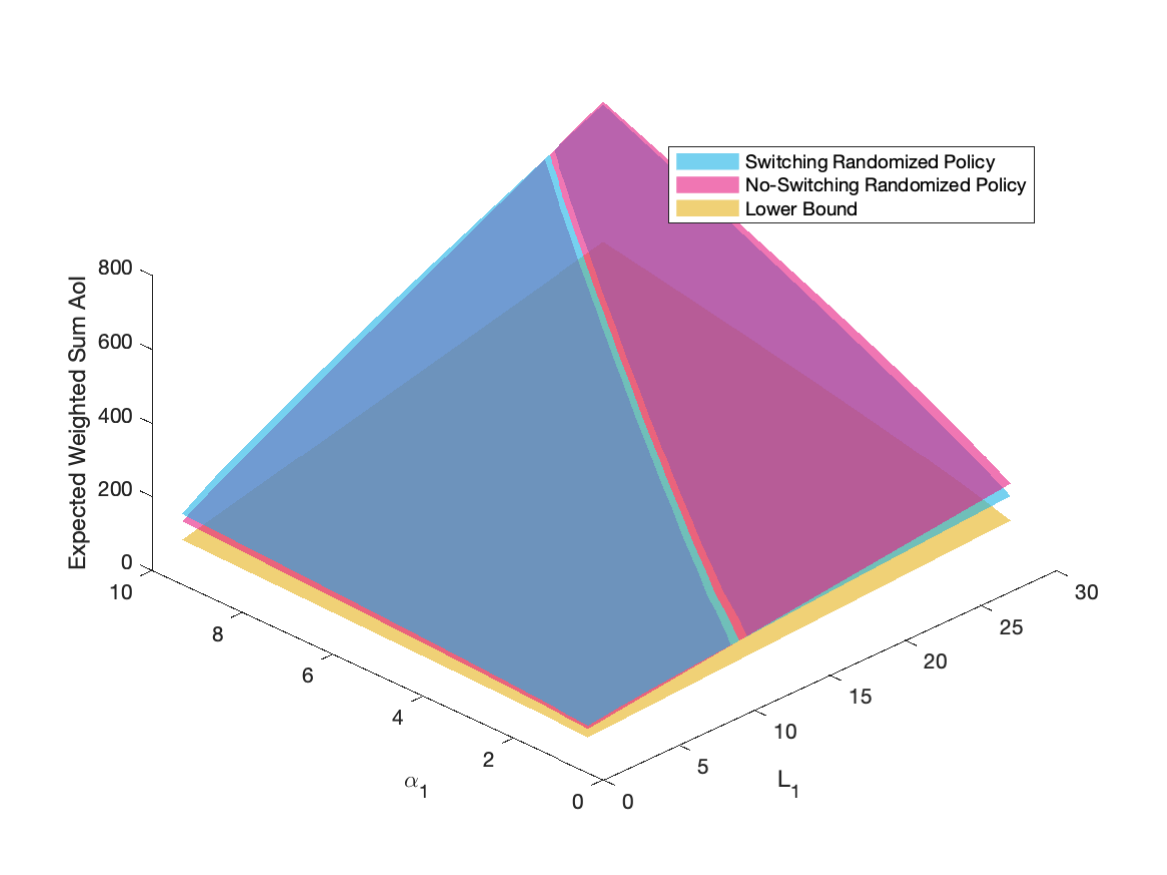}
    \caption{Simulation results of two-source networks with varying weight \(\alpha_1 \in \{1,2,\dots,10\}\) and update length \(L_1 \in \{2,4,\dots,30\}\), while \(\alpha_2 = 10\), \(L_2 = 2\), and $p_1=p_2=0.5$ remain fixed.}\label{fig: randomized}
    \vspace{-2em} 
\end{figure}

We now compare the performance of the optimal SRP and the optimal NSRP in symmetry and non-symmetric networks. 

\begin{corollary}\label{cor: symmetrical randomized}
Consider a symmetric network with channel reliabilities \(p_i = p \in (0,1]\), update lengths \(L_i = L > 0\), and weights \(\alpha_i = \alpha > 0\) for all \(i\). Let \(\mathbb{E}[J^{NS}]\) and \(\mathbb{E}[J^{S}]\) be the EWSAoI achieved by the optimal no-switching and the optimal switching randomized policies, respectively. Then,
\begin{equation}
\mathbb{E}[J^{NS}]\leq \mathbb{E}[J^{S}]
\end{equation}
\end{corollary}



This corollary demonstrates that in symmetric networks the no-switching approach consistently outperforms the switching approach by leveraging continuous transmissions for each source, resulting in lower service times for the selected source, and thus lower EWSAoI.


Figure~\ref{fig: randomized} compares the EWSAoI of the optimal SRP and NSRP in a non-symmetric two-source network with fixed channel reliabilities \(p_1=p_2=0.5\). The parameters for source~1 vary over \(\alpha_1 \in \{1,2,\dots,10\}\) and \(L_1 \in \{2,4,\dots,30\}\), while source~2 has fixed \(\alpha_2=10\) and \(L_2=2\).

\begin{remark}
As can be seen in Fig.~\ref{fig: randomized}, in highly asymmetric networks, the optimal SRP can significantly outperform the optimal NSRP in terms of the EWSAoI. This result is in line with Fig.~\ref{fig:simple_non_sysmmetric} and Table~\ref{table:2sources} which showed switching policies that outperformed the no switching policy by more than $54\%$. 
Intuitively, under a NSRP, once transmission begins for a source with a large update length, the AoI for other sources continues to rise until the update is fully delivered. In these cases, switching to another source with low update length earlier could reduce  the average AoI, as seen in Fig.~\ref{fig:simple_non_sysmmetric}, but the no-switching approach prevents such adaptive flexibility.
\end{remark}

\section{Max-Weight Policy}\label{sec: Max Weight}

In this section, we develop a Max-Weight policy~\cite{neely2022stochastic} designed to reduce the expected drift of a suitably constructed Lyapunov function at every time slot \(t\). The Lyapunov function outputs a nonnegative scalar that is high when the network is in \emph{undesirable} states. 
Prior works including~\cite{kadota2018scheduling,kadota2019minimizing,kadota2019scheduling} utilized Lyapunov functions and one-slot Lyapunov drift analysis that focused on AoI $h_i(t)$ and system times $z_i(t)$. 
While this approach is suitable for networks with $L_i=1,\forall i$ in which every packet transmission may lead to a reduction of AoI in the next time slot. This approach is not suitable for networks with large $L_i$ when the AoI reduction (i.e., the reward) may come in the distant future and may depend on the (stochastic) outcome of future scheduling decisions. This time-dependency and complexity also makes multi-slot Lyapunov drift analysis~\cite{neely2022stochastic} unsuitable. 

To address this challenge, we draw inspiration from Remark~\ref{rem:service_time} to define a novel Lyapunov function that incorporates \emph{waiting time}, ``optimistic'' \emph{service time}, and \emph{throughput debt}~\cite{hou2009theory}, and is amenable to one-slot Lyapunov drift analysis. 
Before developing the Max-Weight policy, we introduce the throughput debt, the proposed Lyapunov function, and the corresponding one-slot Lyapunov drift.

\noindent\textbf{Throughput Debt}. Let $x_i(t)$ denote the throughput debt associated with source $i$ at the beginning of slot $t$. The throughput debt is defined as $x_i(t+1) = t\bar{q}_i-\sum_{\tau=1}^td_i(\tau)$, 
where $\bar{q}_i$ is the long-term throughput target. The value of $t\bar{q}_i$ can be interpreted as the minimum number of packets that source $i$ should have delivered by slot $t + 1$ and $\sum_{\tau=1}^td_i(\tau)$ is the total number of packets actually delivered. 
Let the positive part of the throughput debt be $x^+_i (t) = \max\{x_i(t);0\}$. 
A large debt $x^+_i (t)$ indicates that source $i$ is lagging behind in terms of throughput. Notice that strong stability of the process $x^+_i (t)$, namely 
\begin{equation}
    \lim_{T\rightarrow\infty}\frac{1}{T}\sum_{t=1}^T\mathbb{E}[x^+_i (t)] < \infty
\end{equation}
is sufficient to establish that the long-term throughput is larger than the target, i.e., $q_i \geq \bar{q}_i$.~\cite[Theorem 2.8]{neely2010stability}

\noindent\textbf{Lyapunov Function}. We propose the following Lyapunov function
\begin{equation}\label{lyapunov function}
\begin{aligned}
        \mathcal{L}(t)&=\sum_{i=1}^{N}\beta_{i}\left[h_{i}(t)-z_{i}(t)\right]^2+\sum_{i=1}^{N}\gamma_{i}\left[z_{i}(t)+L_{i}(t)\right]^2\\&+\sum_{i=1}^{N}\frac{V}{2}\left[x_{i}^{+}(t)\right]^{2}
\end{aligned}
\end{equation}
Notice that \(h_i(t) - z_i(t)\) and \(z_i(t) + L_i(t)\) capture the \emph{waiting time} and an ``optimistic'' \emph{service time} of the information update currently in source $i$, respectively. The service time is optimistic as it assumes that all remaining packets will take one slot to be delivered. 
The positive hyper-parameters \(\beta_i\), \(\gamma_i\), and \(V\) are used to tune the Max-Weight policy to different network configurations. From Remark~\ref{rem:service_time}, we know that service time contributes more to the EWSAoI than waiting time, thus, \(\gamma_i\) should be set to a higher value than \(\beta_i\).

\noindent\textbf{One-slot Lyapunov Drift}. Let the network state observed by the BS at the beginning of slot~\(t\) be \(\mathbb{S}(t) := \{h_i(t), z_i(t), L_i(t), x_i(t)\}_{i=1}^N\). The one-slot Lyapunov drift is defined as
\begin{equation}\label{lyapunov Drift}
    \Delta\left(\mathbb{S}(t)\right):=\mathbb{E}\left[\mathcal{L}(t+1)-\mathcal{L}(t)|\mathbb{S}(t)\right]
\end{equation}
%
By substituting the evolution of \(L_i(t)\), \(z_i(t)\), and \(h_i(t)\) from~\eqref{Levolves},~\eqref{zevolves} and~\eqref{hevolves}, respectively, into the drift expression in~\eqref{lyapunov Drift} and performing algebraic manipulations, we obtain an upper bound on \(\Delta\bigl(\mathbb{S}(t)\bigr)\). The resulting bound is expressed in~\eqref{lyapunov drift upper bound1}--\eqref{lyapunov drift upper bound3}, with detailed steps provided in Appendix F.
\begin{equation}\label{lyapunov drift upper bound1}
    \Delta\left(\mathbb{S}(t)\right)\leq B(t)-\sum_{i=1}^{N}p_{i}\mathbb{E}\left[u_{i}(t)|\mathbb{S}(t)\right]C_{i}(t)
\end{equation}
where
\begin{equation}\label{lyapunov drift upper bound2}
\begin{aligned}
B(t)=&\sum_{i=1}^{N}\beta_{i}\mathbb{I}_{L_i(t)=L}\left[2h_{i}(t)-1\right]+V\left[x^+_i(t)\bar{q}_i+\frac{1}{2}\right] \\+&\sum_{i=1}^{N}\gamma_{i}\mathbb{I}_{L_i(t)>1}\left[2\left(z_{i}(t)+L_{i}(t)\right)-1\right],
\end{aligned}
\end{equation}
\begin{equation}\label{lyapunov drift upper bound3}
\begin{aligned}
C_{i}(t) & =\beta_{i}\mathbb{I}_{L_i(t)=L}\left[2h_{i}(t)-1\right]\\&+\beta_{i}\mathbb{I}_{L_i(t)=1}\left[h_{i}^{2}(t)-2h_{i}(t)z_{i}(t)\right]\\
 & +\gamma_{i}\mathbb{I}_{L_i(t)>1}\left[2z_{i}(t)+2L_{i}(t)-1\right]\\&+\gamma_{i}\mathbb{I}_{L_i(t)=1}\left[\left(z_{i}(t)+2\right)^{2}-\left(L_{i}+1\right)^{2}\right]+ Vx^+(t).
\end{aligned}
\end{equation}
The values of $B(t)$ and $C_{i}(t)$ can be easily calculated by any admissible policy and thus can be used for making scheduling decisions in real-time. The remaining update length \(L_i(t)\) appears in~\eqref{lyapunov drift upper bound1}--\eqref{lyapunov drift upper bound3} due to the dependence on evolution of \(z_i(t)\) and \(h_i(t)\) specified in \eqref{zevolves} and \eqref{hevolves}. 

\noindent\textbf{Max-Weight policy.} To minimize the upper bound~\eqref{lyapunov drift upper bound1}, the Max-Weight (MW) policy selects, in each slot $t$, the source with highest value of $C_{i}(t)$, with ties being broken arbitrarily. Intuitively, by minimizing the one-slot Lyapunov drift, the Max-Weight policy will jointly minimize the waiting time and service time, resulting in low EWSAoI. 

Theorem~\ref{theo: Performance Bound for MW} provides a constant factor optimality guarantee for the MW policy. Before introducing Theorem~\ref{theo: Performance Bound for MW}, we define the long-term throughput associated with the lower bound in Theorem~\ref{theo:lower_bound} (for details, please refer to the proof of Theorem~\ref{theo:lower_bound})
%
\begin{equation}
q_i^{L_B} = \frac{\sqrt{\frac{\alpha_i\,L_i\,p_i}{2}}}{\sum_{j=1}^N \sqrt{\frac{\alpha_j\,L_j}{2\,p_j}}}    
\end{equation}
\begin{theorem}\label{theo: Performance Bound for MW}
    For a network with parameters $\{N,\alpha_i,p_i,L_i\}$, by choosing the constants $\beta_{i}=\frac{\alpha_i}{q_i^{L_B}}, \gamma_i=\frac{\alpha_i}{q_i^{L_B}\sqrt{p_i}}$, $V>0$ small enough (see Lemma~\ref{lem:Ci-ordering}), and $\bar{q}_i=q^{L_B}_i-\epsilon$, where $\epsilon\rightarrow0$, the optimality ratio of Max-Weight policy is such that
\begin{equation}\label{eq: optimal ratio MW}
\rho^{MW}=6+\frac{\sqrt{\varPsi}}{NL_B}
\end{equation}
\vspace{-0.4em}
where
\begin{equation}
\begin{aligned}
&\varPsi=\left(\sum_{i=1}^{N}\alpha_i\right)\Biggl(\sum_{i=1}^{N}\frac{\alpha_i}{\sqrt{p_i}}\Bigl[5\frac{L^2_{i}\sqrt{p_i}}{\bar{q}^2_i}-\frac{L_{i}}{\bar{q}_{i}}\\&-\frac{L_{i}\sqrt{p_i}}{\bar{q}_{i}}+2\frac{L^2_i}{\bar{q}^2_i}+2\frac{L_i^2}{\bar{q}_i}+\frac{-L_i^2+8L_i+24}{2}\Bigl]\Biggl)
\end{aligned}
\end{equation}
\end{theorem}

\begin{proof}
First, we prove that if there exists a SRP \(s\) satisfying \(p_i\mu^s_i\geq \bar{q}_i\) for all \(i\), then the MW policy also satisfies the throughput targets \(\{\bar{q}_i\}_{i=1}^N\). 
Next, we perform algebraic manipulations to further bound the inequality in~\eqref{lyapunov drift upper bound1}. We emphasize that, due to the dependency on $L_i(t)$, these algebraic manipulations departed from the traditional Lyapunov drift analysis commonly found in prior works including~\cite{kadota2018scheduling,kadota2019minimizing,kadota2019scheduling}. 
Finally, by comparing MW with the Lower Bound, we obtain~\eqref{eq: optimal ratio MW}.
Detailed derivations are provided in Appendix~\ref{Proof of Theorem theo: Performance Bound for MW}.
\end{proof}

\begin{remark}
The first term in \(\varPsi\) scales as \(\mathcal{O}(N)\), while the second term scales as \(\mathcal{O}(N^3)\) due to \(1/\bar{q}_i\) being \(\mathcal{O}(N)\). Consequently, \(\sqrt{\varPsi}\) scales as \(\mathcal{O}(N^2)\), matching the scaling of \(N L_B\). Hence, the optimality guarantee of the MW policy is bounded by a constant, irrespective of the network size \(N\).
\end{remark}

Numerical results in Sec~\ref{sec:simulation} show that MW outperforms both optimal SRP and optimal NSRP in every network configuration simulated. However, by comparing Theorems~\ref{theo: Performance Bound for MW} and~\ref{theo:stationary_performance}, it might seem that the optimal SRP yields a better performance than MW. This is because the analysis associated with MW is significantly more challenging, leading to an optimality ratio $\rho^{MW}$ that is looser than $\rho^S$. \emph{To the best of our knowledge, these are the first policies with a constant factor optimality guarantee in terms of AoI for networks with different update lengths and unreliable channels.}

\section{Simulation Results}\label{sec:simulation}

In this section, we evaluate the performance of several scheduling policies in terms of their EWSAoI. Specifically, we compare the following policies:
 the optimal SRP proposed in Section~\ref{sec: Switching Randomized Policy}.
 the optimal NSRP proposed in Section~\ref{sec: No-Switching Randomized Policy};
 the Greedy policy in which the BS selects the source with the highest \(h_i(t)\) in each slot $t$;
 the Max-Weight policy for \(L=1\) (MWL1) proposed in~\cite{kadota2018scheduling} in which the BS selects the source with highest value of $\sqrt{\alpha_ip_i}h_i(t)$ in each slot $t$;
 and the Max-Weight policy proposed in Section~\ref{sec: Max Weight}.
Their performance is benchmarked against the lower bound derived in Sec.~\ref{sec:lower_bound}.

\begin{figure}
    \centering
    \includegraphics[trim=0 0 0 8,clip,width=1\linewidth]{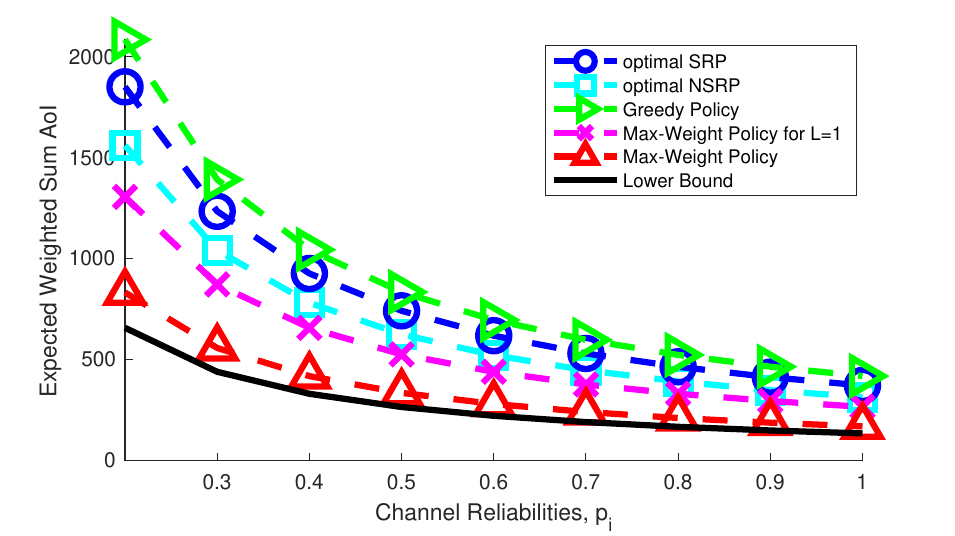}
    \caption{Simulation results for networks with varying channel reliabilities. The network comprises \(N=10\) sources equally divided into Class~1  with \(\alpha_i=5\) and \(L_i=2\) and Class~2  with \(\alpha_i=1\) and \(L_i=50\). Channel reliabilities vary over \(p_i \in \{0.2,\,0.15,\,\ldots,\,1\}\) for all sources.}
    \label{fig:performance varying channel reliabilities}
    \vspace{-1.5em} 
\end{figure}

\begin{figure}
    \centering
    \includegraphics[trim=0 0 0 8,clip,width=1\linewidth]{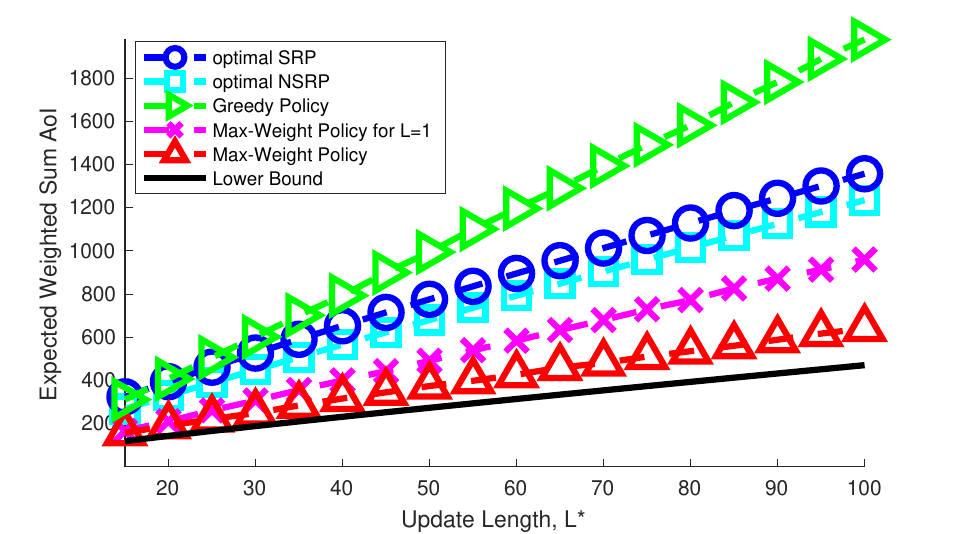}
    \caption{Simulation results for networks with varying update lengths. The network comprises \(N=10\) sources equally divided into Class~1  with \(\alpha_i=5\), \(p_i=0.8\), and \(L_i=2\) and Class~2  with \(\alpha_i=1\) and \(p_i=0.4\). Update lengths for Class~2 are parameterized as \([L^*-2,\,L^*-1,\,L^*,\,L^*+1,\,L^*+2]\) with \(L^* \in \{15,\,20,\,\ldots,\,100\}\).}
    \label{fig:performance varying update length}
    \vspace{-2em}  
\end{figure}

\begin{figure}[h]
    \centering
\includegraphics[trim=0 0 0 10,clip,width=1\linewidth]{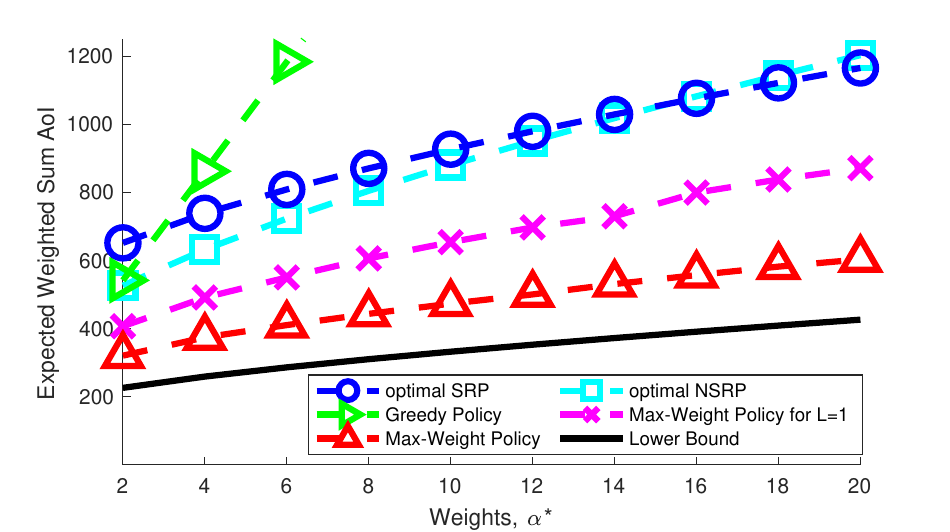}
    \caption{Simulation results for networks with varying weights. The network comprises \(N=10\) sources equally divided into Class~1 with variable priorities \(\alpha_i \in \{2,\,4,\,\ldots,\,20\}\) with fixed \(L=2\) and \(p_i=0.8\), and Class~2 sources have fixed parameters \(\alpha_i=1\), \(L_i=50\), and \(p_i=0.4\).}
\label{fig:performance varying weights}
    \vspace{-1em} 
\end{figure}

\begin{figure}
    \centering
\includegraphics[trim=0 2 0 3,clip,width=1\linewidth]{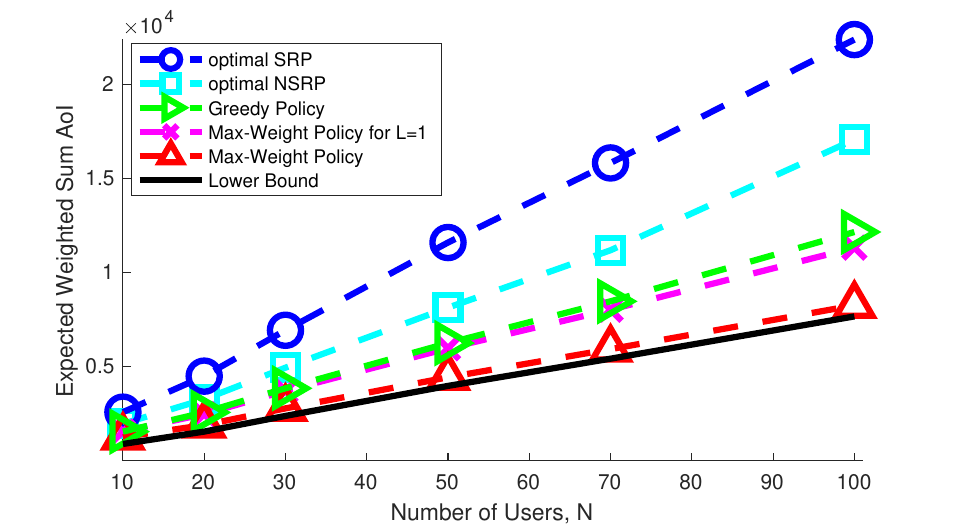}
    \caption{Simulation results for networks with a varying number of sources. The network comprises \(N\) sources with \(N\in\{10,20,30,50,70,100\}\). For each source, the weight is drawn from \(\alpha_i \sim U[1,10]\), the channel reliability from \(p_i \sim U[0.5,1]\), and each source is equally likely to belong to Class~1 (with \(L_i\sim U[2,5]\)) or Class~2 (with \(L_i\sim U[20,100]\)).}
    \label{fig:performance varying number of users}
    \vspace{-2em} 
\end{figure}




We consider two types of sources in our simulations: \emph{Class~1} sources are characterized by high weight \(\alpha_i\) and small update length \(L_i\); and 
\emph{Class~2} sources with low weight \(\alpha_i\) and large update length \(L_i\).

In Figures~\ref{fig:performance varying channel reliabilities}, \ref{fig:performance varying update length}, and \ref{fig:performance varying weights}, we consider networks with \(N=10\) sources, equally divided into five Class~1 and five Class~2 sources. 
In Fig.~\ref{fig:performance varying channel reliabilities}, Class~1 sources are configured with a priority \(\alpha_i = 5\) and a fixed update length \(L_i = 2\), while Class~2 sources have \(\alpha_i = 1\) and \(L_i = 50\); in this setup, the channel reliabilities for all sources vary over the set \(p_i \in \{0.2, 0.15, \ldots, 1\}\). 
In Fig.~\ref{fig:performance varying update length}, the parameters for Class~1 sources remain the same (\(\alpha_i = 5\), \(p_i = 0.8\), \(L_i = 2\)), and for Class~2 sources, while the priority is fixed at \(\alpha_i = 1\) and the channel reliability at \(p_i = 0.4\), the update lengths vary as \([L^*-2,\,L^*-1,\,L^*,\,L^*+1,\,L^*+2]\) with \(L^*\) taking values from \(\{15,\,20,\,\ldots,\,100\}\). 
In Fig.~\ref{fig:performance varying weights}, Class~1 sources have variable priorities \(\alpha_i^* \in \{2,\,4,\,\ldots,\,20\}\) with \(p_i = 0.8\) and \(L_i = 2\), while Class~2 sources have \(\alpha_i = 1\), \(p_i = 0.4\), and \(L_i = 50\).

In Fig.~\ref{fig:performance varying number of users}, we simulate networks with the number of sources \(N\) varying over \(\{10, 20, 30, 50, 70, 100\}\). Moreover, each source is equally likely to belong to either Class~1 or Class~2. Class~1 sources have an update length \(L_i\sim U[2,5]\), while Class~2 sources have \(L_i\sim U[20,100]\). For each source, the weight \(\alpha_i \sim U[1,10]\) and the channel reliability \(p_i \sim U[0.5,1]\). 

Our results clearly demonstrate the superior performance of the Max-Weight policy. Figures~\ref{fig:performance varying channel reliabilities}, \ref{fig:performance varying update length}, and \ref{fig:performance varying weights} show that the Max-Weight policy achieves near-optimal performance across various scenarios. In particular, Max-Weight policy consistently outperforms MWL1 especially in networks with non-symmetric update lengths, we observe that EWSAoI improves by 57$\%$ in Fig~\ref{fig:performance varying channel reliabilities}, 30$\%$ in Fig~\ref{fig:performance varying update length}, and 33$\%$ in Fig.~\ref{fig:performance varying weights} in average. Moreover, Fig.~\ref{fig:performance varying number of users} indicates that the performance of the Max-Weight policy remains robust as the network size increases.


\section{Conclusion}\label{sec: Final Remarks}
In this paper, we considered a single-hop wireless network with a number of nodes transmitting time-sensitive updates to a Base Station over unreliable channels, where each updates consists of multiple packets. We addressed the problem of minimizing the Expected Weighted Sum AoI of the network with large updates. Three low-complexity scheduling policies were developed: optimal SRP, optimal NSRP, and Max-Weight policy. The performance of each policy was evaluated both analytically and through simulation. The Max-Weight policy demonstrated the best performance in terms of AoI. Interesting extensions include consideration of jointly designing the sampling and scheduling algorithm, general non-linear functions of AoI, fairness of AoI among the different sources with different update sizes, and distributed scheduling schemes.

\bibliographystyle{IEEEtran}
\bibliography{references}

\begin{thebibliography}{10}
\providecommand{\url}[1]{#1}
\csname url@samestyle\endcsname
\providecommand{\newblock}{\relax}
\providecommand{\bibinfo}[2]{#2}
\providecommand{\BIBentrySTDinterwordspacing}{\spaceskip=0pt\relax}
\providecommand{\BIBentryALTinterwordstretchfactor}{4}
\providecommand{\BIBentryALTinterwordspacing}{\spaceskip=\fontdimen2\font plus
\BIBentryALTinterwordstretchfactor\fontdimen3\font minus \fontdimen4\font\relax}
\providecommand{\BIBforeignlanguage}[2]{{%
\expandafter\ifx\csname l@#1\endcsname\relax
\typeout{** WARNING: IEEEtran.bst: No hyphenation pattern has been}%
\typeout{** loaded for the language `#1'. Using the pattern for}%
\typeout{** the default language instead.}%
\else
\language=\csname l@#1\endcsname
\fi
#2}}
\providecommand{\BIBdecl}{\relax}
\BIBdecl

\bibitem{kaul2011minimizing}
S.~Kaul \emph{et~al.}, ``Minimizing age of information in vehicular networks,'' in \emph{IEEE SECON}, 2011, pp. 350--358.

\bibitem{AoI_V}
C.~Guo \emph{et~al.}, ``Age of information, latency, and reliability in intelligent vehicular networks,'' \emph{IEEE Network}, vol.~37, no.~6, pp. 109--116, 2023.

\bibitem{AoIUAV1}
H.~Hu \emph{et~al.}, ``{AoI}-minimal trajectory planning and data collection in uav-assisted wireless powered {IoT} networks,'' \emph{IEEE Internet of Things Journal}, vol.~8, no.~2, pp. 1211--1223, 2021.

\bibitem{9637803}
B.~Choudhury \emph{et~al.}, ``{AoI}-minimizing scheduling in uav-relayed {IoT} networks,'' in \emph{Proc. IEEE MASS}, 2021, pp. 117--126.

\bibitem{WiSwarm}
V.~Tripathi \emph{et~al.}, ``{WiSwarm}: Age-of-information-based wireless networking for collaborative teams of {UAVs},'' in \emph{Proc. IEEE INFOCOM}, 2023, pp. 1--10.

\bibitem{yu2022age}
B.~Yu, X.~Chen, and Y.~Cai, ``Age of information for the cellular internet of things: Challenges, key techniques, and future trends,'' \emph{IEEE Communications Magazine}, vol.~60, no.~12, pp. 20--26, 2022.

\bibitem{beytur2020towards}
H.~B. Beytur \emph{et~al.}, ``Towards {AoI}-aware smart {IoT} systems,'' in \emph{Proc. IEEE ICNC}, 2020, pp. 353--357.

\bibitem{abd2019role}
M.~A. Abd-Elmagid \emph{et~al.}, ``On the role of age of information in the internet of things,'' \emph{IEEE Communications Magazine}, vol.~57, no.~12, pp. 72--77, 2019.

\bibitem{kadota2019minimizing}
I.~Kadota and E.~Modiano, ``Minimizing the age of information in wireless networks with stochastic arrivals,'' \emph{IEEE Transactions on Mobile Computing}, vol.~20, no.~3, pp. 1173--1185, 2019.

\bibitem{zakeri2023minimizing}
A.~Zakeri \emph{et~al.}, ``Minimizing the {AoI} in resource-constrained multi-source relaying systems: Dynamic and learning-based scheduling,'' \emph{IEEE Transactions on Wireless Communications}, vol.~23, no.~1, pp. 450--466, 2023.

\bibitem{zhou2019joint}
B.~Zhou and W.~Saad, ``Joint status sampling and updating for minimizing age of information in the internet of things,'' \emph{IEEE Transactions on Communications}, vol.~67, no.~11, pp. 7468--7482, 2019.

\bibitem{tang2020minimizing}
H.~Tang \emph{et~al.}, ``Minimizing age of information with power constraints: Multi-user opportunistic scheduling in multi-state time-varying channels,'' \emph{IEEE Journal on Selected Areas in Communications}, vol.~38, no.~5, pp. 854--868, 2020.

\bibitem{kadota2019scheduling}
I.~Kadota, A.~Sinha, and E.~Modiano, ``Scheduling algorithms for optimizing age of information in wireless networks with throughput constraints,'' \emph{IEEE/ACM Transactions on Networking}, vol.~27, no.~4, pp. 1359--1372, 2019.

\bibitem{fountoulakis2023scheduling}
E.~Fountoulakis \emph{et~al.}, ``Scheduling policies for {AoI} minimization with timely throughput constraints,'' \emph{IEEE Transactions on Communications}, vol.~71, no.~7, pp. 3905--3917, 2023.

\bibitem{atay2021aging}
E.~U. Atay, I.~Kadota, and E.~Modiano, ``Aging wireless bandits: Regret analysis and order-optimal learning algorithm,'' in \emph{Proc. WiOpt}, 2021.

\bibitem{liu2024optimizing}
J.~Liu, Q.~Wang, and H.~Chen, ``Optimizing information freshness in uplink multiuser mimo networks with partial observations,'' \emph{arXiv preprint arXiv:2401.02218}, 2024.

\bibitem{zhao2025optimizing}
Z.~Zhao and I.~Kadota, ``Optimizing age of information without knowing the age of information,'' in \emph{Proc. IEEE INFOCOM}, 2025.

\bibitem{kadota2018scheduling}
I.~Kadota \emph{et~al.}, ``Scheduling policies for minimizing age of information in broadcast wireless networks,'' \emph{IEEE/ACM Transactions on Networking}, vol.~26, no.~6, pp. 2637--2650, 2018.

\bibitem{li2019minimizing}
C.~Li \emph{et~al.}, ``Minimizing age of information under general models for iot data collection,'' \emph{IEEE Transactions on Network Science and Engineering}, vol.~7, no.~4, pp. 2256--2270, 2019.

\bibitem{tripathi2021computation}
V.~Tripathi \emph{et~al.}, ``Computation and communication co-design for real-time monitoring and control in multi-agent systems,'' in \emph{Proc. WiOpt}, 2021.

\bibitem{zhou2019minimum}
B.~Zhou and W.~Saad, ``Minimum age of information in the internet of things with non-uniform status packet sizes,'' \emph{IEEE Transactions on Wireless Communications}, vol.~19, no.~3, pp. 1933--1947, 2019.

\bibitem{AoI5G}
C.~Li \emph{et~al.}, ``Minimizing {AoI} in a {5G}-based {IoT} network under varying channel conditions,'' \emph{IEEE Internet of Things Journal}, vol.~8, no.~19, pp. 14\,543--14\,558, 2021.

\bibitem{ross2014introduction}
S.~M. Ross, \emph{Introduction to probability models}.\hskip 1em plus 0.5em minus 0.4em\relax Academic press, 2014.

\bibitem{neely2022stochastic}
M.~Neely, \emph{Stochastic network optimization with application to communication and queueing systems}.\hskip 1em plus 0.5em minus 0.4em\relax Springer Nature, 2022.

\bibitem{hou2009theory}
I.-H. Hou, V.~Borkar, and P.~R. Kumar, ``A theory of {QoS} for wireless,'' in \emph{Proc. IEEE INFOCOM}, 2009.

\bibitem{neely2010stability}
M.~J. Neely, ``Stability and capacity regions or discrete time queueing networks,'' \emph{arXiv preprint arXiv:1003.3396}, 2010.

\bibitem{stochatic}
R.~G. Gallager, \emph{Stochastic processes: theory for applications}.\hskip 1em plus 0.5em minus 0.4em\relax Cambridge University Press, 2013.

\end{thebibliography}

\onecolumn  
\appendices

\section{Proof of Propsition~\ref{prop: AoI expression with waiting/service time}}\label{Proof of Propsition prop: AoI expression with waiting/service time}

\textbf{Propsition~\ref{prop: AoI expression with waiting/service time}. }
The infinite-horizon Weighted Sum AoI achieved by scheduling policy $\pi$, namely \(J^\pi\), can be written as
\begin{equation}
J^\pi = \lim_{T\rightarrow \infty }\sum_{i=1}^{N}\frac{\alpha_i}{N}
\Biggl[
\frac{\Bar{\mathbb{M}}\bigl[(W_i[m]+S_i[m])^2\bigr]}{2\,\Bar{\mathbb{M}}\bigl[W_i[m]+S_i[m]\bigr]} 
+\frac{\Bar{\mathbb{M}}\bigl[S_{i}[m-1]\,(W_i[m]+S_i[m])\bigr]}
{\Bar{\mathbb{M}}\bigl[W_i[m]+S_i[m]\bigr]}+1
\Biggr],
\end{equation}
where \(W_i[m]\) and \(S_i[m]\) are the waiting time and service time of the \(m\)th update for destination \(i\).

\begin{proof}

\begin{figure}[h]
    \centering
    \includegraphics[width=0.8\linewidth]{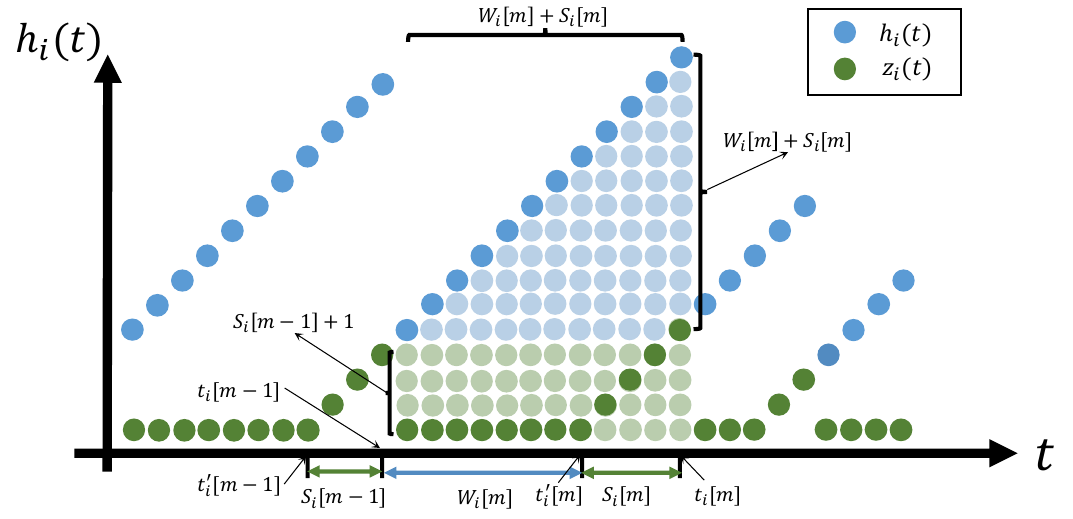}
    \caption{A sample path of AoI Evolution.}
    \label{fig: proof of expression}
\end{figure}

Consider a network operating under policy \(\pi\) over a time horizon \(T\). Let \(\Omega\) be the associated sample space, and let \(\omega \in \Omega\) denote a sample path, as shown in Fig.~\ref{fig: proof of expression}. Recall that \(K_i(T)\) is the total number of updates delivered to destination \(i\) by the end of slot \(T\), and the \emph{waiting time} and \emph{service time} of the \(m\)th update for source \(i\) are given by
\begin{equation}
    W_i[m] \;:=\; t'_i[m] - t_i[m-1],
\end{equation}
\begin{equation}
    S_i[m] \;:=\; t_i[m] - t'_i[m].
\end{equation}
Denote by $R_i$ be the number of slots remaining after the last update delivery. Then, the time-horizon can be written as follows
\begin{equation}\label{eq: T=W+S}
    T=\sum_{m=1}^{K_i(T)} \left(W_i[m]+S_i[m]\right)+R_i, 
\end{equation}

The evolution of $h_i(t)$ is well-defined in each of the time intervals $W_i[m]$, $S_i[m]$, and $R_i$ . According to~\eqref{hevolves}, during the interval$\big[t_i[m-1]+1,t_i[m]\big]$, the parameter $h_i (t)$ evolves as $\{S_i [m-1] + 2,S_i[m-1] + 3,... ,S_i[m-1] + W_i[m]+S_i[m]+1\}$. This pattern is repeated throughout the entire time-horizon, for $m \in \{1,2,...,K_i(T)\}$, and also during the last $R_i$ slots. As a result, the time-average AoI associated with destination $i$ can be expressed as
\begin{equation}\label{AoI average finite T}
    \frac{1}{T}\sum_{t=1}^{T}h_{i}(t)=\frac{1}{T}\left[\sum_{m=1}^{D_{i}(T)}\frac{\left(W_{i}[m]+S_{i}[m]\right)^{2}}{2}+\left(S_{i}[m-1]+1\right)\left(W_{i}[m]+S_{i}[m]\right)\right]+\frac{1}{T}\left[\left(S_{i}[K_i(T)]+1\right)R_i+R^2_i\right]
\end{equation}

Combining~\eqref{eq: T=W+S} with the sample mean $\Bar{\mathbb{M}} \bigl[W_i[m]\bigl] $ and $\Bar{\mathbb{M}} \bigl[S_i[m]\bigl]$, yields
\begin{equation}\label{T/D=MW+MS}
    \frac{T}{K_i(T)} = \Bar{\mathbb{M}} \bigl[W_i[m]\bigl] +  \Bar{\mathbb{M}} \bigl[S_i[m]\bigl]+\frac{R_i}{K_i(T)}
\end{equation}
Substituting~\eqref{T/D=MW+MS} into~\eqref{AoI average finite T} yields
\begin{equation}
    \begin{aligned}\label{AoI average finite T with sample mean}
    \frac{1}{T}\sum_{t=1}^{T}h_{i}(t)=&\Biggl[
\frac{\frac{\Bar{\mathbb{M}}\bigl[(W_i[m]+S_i[m])^2\bigr]}{2}+\Bar{\mathbb{M}}\bigl[S_{i}[m-1](W_i[m]+S_i[m])\bigr]+\Bar{\mathbb{M}}\bigl[W_i[m]\bigr]+\Bar{\mathbb{M}}\bigl[S_i[m]\bigr]}
{\Bar{\mathbb{M}}\bigl[W_i[m]\bigr]+\Bar{\mathbb{M}}\bigl[S_i[m]\bigr]+\frac{R_i}{K_i(T)}}\\&+\frac{\left(S_{i}[K_i(T)]+1\right)R_i+R^2_i}{\Bar{\mathbb{M}}\bigl[W_i[m]\bigr]+\Bar{\mathbb{M}}\bigl[S_i[m]\bigr]+\frac{R_i}{K_i(T)}}
\Biggr],
\end{aligned}
\end{equation}

The next step is to take the limit of~\eqref{AoI average finite T with sample mean} as $T\rightarrow\infty$. Without loss of generality, we assume that $R_i<\infty$, gives

\begin{equation}
\lim_{T\rightarrow \infty }\frac{1}{T}\sum_{t=1}^{T}h_{i}(t) = 
\lim_{T\rightarrow \infty }\Biggl[
\frac{\Bar{\mathbb{M}}\bigl[(W_i[m]+S_i[m])^2\bigr]}{2\,\Bar{\mathbb{M}}\bigl[W_i[m]+S_i[m]\bigr]} 
+\frac{\Bar{\mathbb{M}}\bigl[S_{i}[m-1]\,(W_i[m]+S_i[m])\bigr]}
{\Bar{\mathbb{M}}\bigl[W_i[m]+S_i[m]\bigr]}+1
\Biggr],
\end{equation}

Taking the weighted sum average across the sources with respected to weights $\alpha_i$, yields
\begin{equation}
J^\pi = \lim_{T\rightarrow \infty }\sum_{i=1}^{N}\frac{\alpha_i}{N}
\Biggl[
\frac{\Bar{\mathbb{M}}\bigl[(W_i[m]+S_i[m])^2\bigr]}{2\,\Bar{\mathbb{M}}\bigl[W_i[m]+S_i[m]\bigr]} 
+\frac{\Bar{\mathbb{M}}\bigl[S_{i}[m-1]\,(W_i[m]+S_i[m])\bigr]}
{\Bar{\mathbb{M}}\bigl[W_i[m]+S_i[m]\bigr]}+1
\Biggr],
\end{equation}

\end{proof}

\section{Proof of Theorom~\ref{theo:lower_bound}}
\label{Proof of Theorom theo:lower_bound}
\textbf{Theorem~\ref{theo:lower_bound}. }
For a network with parameters \(\{N,\alpha_i,p_i,L_i\}\), the following bound holds for all admissible policies \(\pi \in \Pi\):
\begin{equation}
L_{B} =\frac{1}{2}\Biggl(\sum_{i=1}^{N}\sqrt{\frac{\alpha_i\,L_{i}}{p_{i}}}\Biggr)^{2}+\sum_{i=1}^{N}\alpha_i
\leq\mathbb{E}[J^\pi],
\end{equation}
\begin{proof}

Applying Jensen Inequality $\Bar{\mathbb{M}}\bigl[(W_i[m]+S_i[m])^2\bigr]\geq \Bar{\mathbb{M}}\bigl[W_i[m]+S_i[m]\bigr]^2$ and $S_{i}[m-1](W_i[m]+S_i[m])>0, \forall i,m,$ into~\eqref{performance for waiting and service time} yields

\begin{equation}\label{j geq ws}
J^\pi > \lim_{T\rightarrow \infty }\sum_{i=1}^{N}\frac{\alpha_i}{N}
\Biggl[
\frac{\Bar{\mathbb{M}}\bigl[(W_i[m]+S_i[m])\bigr]}{2} 
+1
\Biggr],
\end{equation}

Notice that an information update delivery includes $L_i$ data packet delivery, hence, we can rewrite the total number of delivered packets $D_i(T)$ as
\begin{equation}
    D_i(T)= L_iK_i(T)+O_i
\end{equation}
Where $O_i$ is the packet delivered after delivery of update $K_i(T)$, upper bounded by $O_i<L_i$.
Consider the infinite time-horizon $T\rightarrow\infty$ and leverage~\eqref{T/D=MW+MS} in Appendix~\ref{Proof of Propsition prop: AoI expression with waiting/service time}, we rewrite the long-term data packet throughput under policy $\pi $ defined in~\eqref{throughput define} as 
\begin{equation}\label{q=1/W+S}
    q_{i}^{\pi}= \lim_{T\rightarrow \infty }\frac{L_iK_i(T)+O_i}{T}= \frac{L_i}{\Bar{\mathbb{M}} \bigl[W_i[m]\bigl] +  \Bar{\mathbb{M}} \bigl[S_i[m]\bigl]}
\end{equation}

Substituting~\eqref{q=1/W+S} into~\eqref{j geq ws} yields

\begin{equation}
    J^\pi > \frac{1}{N}\sum_{i=1}^{N}\alpha_i\left(\frac{L_{i}}{2q_{i}^{\pi}}+1\right)
\end{equation}

Therefore, the lower bound can be obtained by sloving the optimization problem:
\begin{equation}\label{L_B opt 1}
    L_{B}=\min_{\pi\in\Pi}\frac{1}{N}\sum_{i=1}^{N}\alpha_i\left(\frac{L_{i}}{2q_{i}^{\pi}}+1\right)
\end{equation}
\begin{equation}\label{L_B opt 2}
    s.t.\sum_{i=1}^{N}\frac{q_{i}^{\pi}}{p_{i}}\leq1
\end{equation}

Substituting~\eqref{L_B opt 2} into following Cauchy-Schwarz inequality 
\begin{equation}
    \frac{1}{2}\Biggl(\sum_{i=1}^{N}\sqrt{\frac{\alpha_i\,L_{i}}{p_{i}}}\Biggr)^{2}\leq\left(\sum_{i=1}^{N}\frac{\alpha_iL_{i}}{2q_{i}^{\pi}}\right)\left(\sum_{i=1}^{N}\frac{q_{i}^{\pi}}{p_{i}}\right)
\end{equation}
where equation holds when
\begin{equation}
q_i^{L_B} = \frac{\sqrt{\frac{\alpha_i\,L_i\,p_i}{2}}}{\sum_{j=1}^N \sqrt{\frac{\alpha_j\,L_j}{2\,p_j}}}    
\end{equation}

and applying into~\eqref{L_B opt 1} yields
\begin{equation}
L_{B} =\frac{1}{2}\Biggl(\sum_{i=1}^{N}\sqrt{\frac{\alpha_i\,L_{i}}{p_{i}}}\Biggr)^{2}+\sum_{i=1}^{N}\alpha_i
\end{equation}

\end{proof}


\section{Proof of Proposition~\ref{prop:EWSAoIRandomzedproposition}}
\label{Proof of Theorom prop:EWSAoIRandomzedproposition}
\textbf{Proposition~\ref{prop:EWSAoIRandomzedproposition}. }
For any network with parameters \(\{N,\alpha_i,p_i,L_i\}\) and any SRP \(s \in \Pi^s_r\) characterized by \(\{\mu^s_i\}_{i=1}^N\), the EWSAoI is given by
\begin{equation}
\mathbb{E}[J^s] = \frac{1}{N} \sum_{i=1}^N \alpha_i \left( \frac{3 L_i - 1}{2 p_i \mu_i } + 1\right).
\end{equation}

\begin{proof}
Taking the expectation of~\eqref{performance for waiting and service time} over sample path space $\Omega$, yields
\begin{equation}\label{performance for waiting and service time appc}
\mathbb{E}[J^\pi] = \sum_{i=1}^{N}\frac{\alpha_i}{N}
\Biggl[
\frac{{\mathbb{E}}\bigl[(W_i[m]+S_i[m])^2\bigr]}{2{\mathbb{E}}\bigl[W_i[m]+S_i[m]\bigr]} 
+\frac{{\mathbb{E}}\bigl[S_{i}[m-1]\,(W_i[m]+S_i[m])\bigr]}
{{\mathbb{E}}\bigl[W_i[m]+S_i[m]\bigr]}+1
\Biggr],
\end{equation}

For Network employing SRP, since the channel state and scheduling decision during the current packet transmission are independent with history information, the waiting time $W_i[m]$ and service time $S_i[m]$ are independent.
Thus, we omit the update index in~\eqref{performance for waiting and service time appc}, and rewrite as
\begin{equation}\label{E[j]=E[W.S.]}
\mathbb{E}[J^s] = \sum_{i=1}^{N}\frac{\alpha_i}{N}
\Biggl[
\frac{\mathbb{E}\bigl[W^2_i\bigr]+2\mathbb{E}\bigl[W_i\bigr]\mathbb{E}\bigl[S_i\bigr]+\mathbb{E}\bigl[S^2_i\bigr]}{2\mathbb{E}\bigl[W_{i}\bigr]+2\mathbb{E}\bigl[S_{i}\bigr]} 
+\frac{{\mathbb{E}}\bigl[S_{i}\bigr]\Bigl[\mathbb{E}\bigl[W_{i}\bigr]+\mathbb{E}\bigl[S_{i}\bigr]\Bigl]}
{\mathbb{E}\bigl[W_{i}\bigr]+\mathbb{E}\bigl[S_{i}\bigr]}+1
\Biggr],
\end{equation}

The distributions of waiting time and service time for a network employing the SRP $s$ follows negative binomial distribution~\cite{ross2014introduction}. Specifically,  the waiting time follows $W_i\sim NB(1,p_i\mu^s_i)$ and the service time follows $S_i\sim NB(L_i-1,p_i\mu^s_i)$, 
first-order moments and second-order moments of $W_i$ and $S_i$ are given by
\begin{equation}\label{E[S^s]}
    \mathbb{E}\bigl[S_{i}\bigr]=\frac{L_i-1}{p_i\mu^s_i},
\end{equation}
\begin{equation}
    \mathbb{E}\bigl[S^2_{i}\bigr]=\frac{(L_i-1)(L_i-p_i\mu^s_i)}{p^2_i(\mu^s_i)^2},
\end{equation}
\begin{equation}
    \mathbb{E}\bigl[W_{i}\bigr]=\frac{1}{p_i\mu^s_i},
\end{equation}
\begin{equation}\label{E[W^2^s]}
    \mathbb{E}\bigl[W^2_{i}\bigr]=\frac{2-p_i\mu^s_i}{p^2_i(\mu^s_i)^2}.
\end{equation}

Substituting~\eqref{E[S^s]}--\eqref{E[W^2^s]} into~\eqref{E[j]=E[W.S.]} yields
\begin{equation}
\mathbb{E}[J^s] = \frac{1}{N} \sum_{i=1}^N \alpha_i \left( \frac{3 L_i - 1}{2 p_i \mu^s_i } +1 \right).
\end{equation}

\end{proof}

\section{Proof of Theorem~\ref{theo:stationary_performance}}
\label{Proof of Theorom theo:stationary_performance}

\textbf{Theorem~\ref{theo:stationary_performance}. }
For a network with parameters \(\{N,\alpha_i,p_i,L_i\}\), let \(S\in \Pi^s_r\) be the optimal SRP. Its scheduling probabilities \(\{\mu^S_i\}_{i=1}^N\) are given by
\begin{equation}
    \mu^S_i 
    \;=\; 
    \frac{\sqrt{\frac{\alpha_i\,(3\,L_i -1)}{2\,p_i}}}
         {\sum_{j=1}^N \sqrt{\frac{\alpha_j\,(3\,L_j -1)}{2\,p_j}}},
    \quad 
    \forall i.
\end{equation}
The associated EWSAoI is given by 
\begin{equation}
    \mathbb{E}[J^S] 
    \;=\; 
    \frac{1}{N}  \sum_{i=1}^N \alpha_i
    \;+\;
    \frac{1}{N} 
    \left(\sum_{i=1}^N \sqrt{\tfrac{\alpha_i\,(3\,L_i-1)}{2\,p_i}}\right)^2.
\end{equation}
which satisfies
\begin{equation}
    L_B \;\le\; \mathbb{E}[J^S] \;\le\; \rho^S\,L_B,
\end{equation}
where \(L_B\) is the lower bound from Theorem~\ref{theo:lower_bound} and the optimality ratio \(\rho^S\) is
\begin{equation}
    \rho^S=\frac{\frac{1}{N}\Bigl(\sum_{i=1}^N \sqrt{\tfrac{\alpha_i(3L_i-1)}{2p_i}}\Bigr)^2+\frac{1}{N}\sum_{i=1}^N \alpha_i}{\frac{1}{2N}(\sum_{i=1}^{N}\sqrt{\frac{\alpha_iL_{i}}{p_{i}}})^{2}+\frac{1}{N}\sum_{i=1}^{N}\alpha_i}
\end{equation}

\begin{proof}

From~\eqref{eq:EJ^s}, the optimal SRP $S$ follows by solving
\begin{align}
\min_{\{\mu^s_i\}_{i=1}^N}\sum_{i=1}^N  \alpha_i\biggl(\frac{3\,L_i -1}{2\,p_i\,\mu^s_i}\biggr)
\quad \text{s.t.}\quad 
\sum_{i=1}^N \mu^s_i \;\le\; 1.
\end{align}

Consider the following Cauchy-Schwarz  inequality
\begin{equation}
    \frac{1}{2}\Biggl(\sum_{i=1}^{N}\sqrt{\frac{\alpha_i(3L_{i}-1)}{p_{i}}}\Biggr)^{2}\leq\left(\sum_{i=1}^{N}\frac{\alpha_i(3L_{i}-1)}{2p_i\mu^s_i}\right)\left(\sum_{i=1}^{N}\mu^s_i\right)
\end{equation}
where equation holds when 
\begin{equation}
    \mu^S_i =
    \frac{\sqrt{\frac{\alpha_i\,(3\,L_i -1)}{2\,p_i}}}
         {\sum_{j=1}^N \sqrt{\frac{\alpha_j\,(3\,L_j -1)}{2\,p_j}}},
    \quad 
    \forall i.
\end{equation}
Therefore, the optimal SRP $S$ is given by $\{\mu^S_i\}_{i=1}^N$. Substituting $\{\mu^S_i\}_{i=1}^N$ into~\eqref{eq:EJ^s} yields 
\begin{equation}
    \mathbb{E}[J^S] 
    \;=\; 
    \frac{\sum_{i=1}^N \alpha_i}{2\,N}  
    \;+\;
    \frac{1}{N} 
    \Bigl(\sum_{i=1}^N \sqrt{\tfrac{\alpha_i\,(3\,L_i-1)}{2\,p_i}}\Bigr)^2.
\end{equation}
Thus,  the optimal radio of optimal SRP $S$ is given by
\begin{equation}
    \rho^S=\frac{\frac{1}{N}\Bigl(\sum_{i=1}^N \sqrt{\tfrac{\alpha_i(3L_i-1)}{2p_i}}\Bigr)^2+\frac{1}{N}\sum_{i=1}^N \alpha_i}{\frac{1}{2N}(\sum_{i=1}^{N}\sqrt{\frac{\alpha_iL_{i}}{p_{i}}})^{2}+\frac{1}{N}\sum_{i=1}^{N}\alpha_i}
\end{equation}
\end{proof}

\section{Proof of Proposition~\ref{prop:EWSAoIRandomzedpropositionNS}}
\label{Proof of Theorom prop:EWSAoIRandomzedpropositionNS}

\textbf{Proposition~\ref{prop:EWSAoIRandomzedpropositionNS}. }
For any given network model with parameters \(\{N,\alpha_i,p_i,L_i\}\) 
and an arbitrary NSRP \(ns \in \Pi^{ns}_r\) with scheduling 
probabilities \(\{\mu^{ns}_i\}_{i=1}^N\),
the EWSAoI is given by:
\begin{equation}
\begin{aligned}
  \mathbb{E}[J^{ns}] 
=  \sum_{i=1}^N \frac{\alpha_i}{N} \Biggl[\frac{\mu^{ns}_i p_i}{\sum^N_{j=1}\mu^{ns}_j L_j}\Biggl(\frac{\mathbb{E}[W^2_i] + \mathbb{E}[S^2_i]}{2} 
+ \left(\frac{L_i-1}{p_i}\right)^2
+ 2\mathbb{E}[W^2_i]
      \left(\frac{L_i-1}{p_i}\right)\Biggl) +1
  \Biggr]. 
\end{aligned}
\end{equation}
Here, \(\mathbb{E}[S_i^2]\), \(\mathbb{E}[W_i]\), and \(\mathbb{E}[W_i^2]\) denote the second moment of the service time, the first moment of the waiting time, and the second moment of the waiting time, respectively. The term \(\mathbb{E}[Y_i^2]\) represents the second moment of the interval between the time source \(i\) is selected and next selection. These quantities are given by
\begin{align}
    \mathbb{E}[S^2_i] &= \frac{(L_i - 1)(L_i - p_i)}{p_i^2}, \\
    \mathbb{E}[W_i]&=\frac{\sum^N_{j=1}\mu^{ns}_j L_j - \mu^{ns}_i(L_i-1)}{\mu^{ns}_i p_i},\\
    \mathbb{E}[W_i^2]& = \frac{1}{\mu^{ns}_i p_i}\Biggl[\mu^{ns}_i\left(1 + 2(1 - p_i)\mathbb{E}[W_i]\right) + \sum_{j \neq i}\mu^{ns}_j\left(\mathbb{E}[Y_j^2] + 2L_jX_i\right)\Biggl], \\
   \mathbb{E}[Y_i^2]&= 2L_i -1 + \frac{(L_i -1)(L_i - p_i)}{p_i}.
\end{align}
\begin{proof}
    
For Network employing NSRPs, since the channel state and scheduling decision during the current update transmission are independent with history information, the waiting time $W_i[m]$ and service time $S_i[m]$ are independent across different update. Also, within the delivery of an update $m$, the waiting time $W_i[m]$ and service time $S_i[m]$ are independent.  
Thus, taking the expectation of ~\eqref{performance for waiting and service time} associate with networks employing NSRPs, and omitting the update index yields
\begin{equation}\label{E[j^ns]=E[W.S.]}
\mathbb{E}[J^{ns}] = \sum_{i=1}^{N}\frac{\alpha_i}{N}
\Biggl[
\frac{\mathbb{E}\bigl[W^2_i\bigr]+2\mathbb{E}\bigl[W_i\bigr]\mathbb{E}\bigl[S_i\bigr]+\mathbb{E}\bigl[S^2_i\bigr]}{2\mathbb{E}\bigl[W_{i}\bigr]+2\mathbb{E}\bigl[S_{i}\bigr]} 
+\frac{{\mathbb{E}}\bigl[S_{i}\bigr]\Bigl[\mathbb{E}\bigl[W_{i}\bigr]+\mathbb{E}\bigl[S_{i}\bigr]\Bigl]}
{\mathbb{E}\bigl[W_{i}\bigr]+\mathbb{E}\bigl[S_{i}\bigr]}+1
\Biggr],
\end{equation}

Notice that the expression is similar to~\eqref{E[j]=E[W.S.]} in Appednix~\ref{Proof of Theorom prop:EWSAoIRandomzedproposition}, as we use the same technique.
The distributions of service time for a network employing the NSRPs $ns$ also follows negative binomial distribution~\cite{ross2014introduction}, namely $S_i\sim NB(L_i-1,p_i)$, and
the first-order moment and second-order moment of $S_i$ are given by
\begin{equation}\label{E[S^ns]}
    \mathbb{E}\bigl[S_{i}\bigr]=\frac{L_i-1}{p_i},
\end{equation}
\begin{equation}\label{E[S^ns^2]}
\mathbb{E}\bigl[S^2_{i}\bigr]=\frac{(L_i-1)(L_i-p_i)}{p^2_i},
\end{equation}

However, due to the correlation of scheduling decision for different sources, obtaining the distribution of waiting time is challenging. 
Let \((i, L_i(t))\) represent the system state st slot $t$ when $i$ is selected with remaining update length $L_i(t)$, and \(Selecting\) when the BS is randomly selecting sources. 
Consider the Markov chain associated with a network employing NSRPs, and the state transfer diagram for  is illustrated in Fig.~\ref{fig: markov chain}.
Thus, the waiting time $W_i$ is the time duration that from the time slot when the system state transfer from $(i,1)$ to $Selecting$, to the time slot when the system state transfer from $Selecting$ to $(i,L_i-1)$. Furthermore, due to the memoryless property of Markov Chain, the distribution of waiting time $W_i$ is the same the distribution of first passage time from state $Selecting$ to $(i,L_i-1)$~\cite{stochatic}. Next, we leverage the memoryless property and the first passage time argument to derive the first-order moment and second-order moment of the waiting time.

\begin{figure}[h]
    \centering
\includegraphics[width=0.6\linewidth]{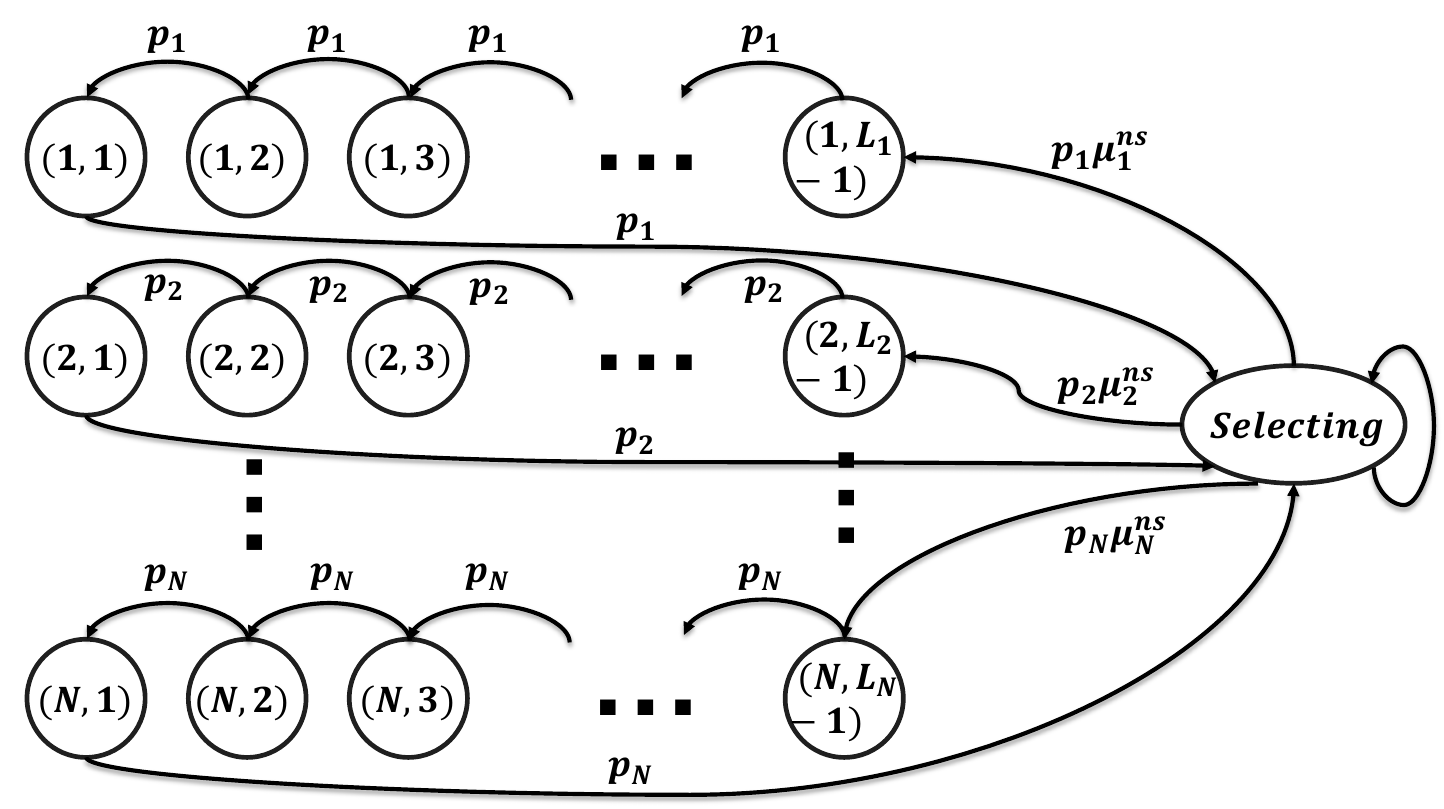}
    \caption{State Transfer Diagram for networks employing NSRPssource $i$}
    \label{fig: markov chain}  
\end{figure}

Suppose that the   system state is $Selecting$ at time slot $t$, there are two categories of scheduling decisions can be made:
\begin{itemize}
    \item Selecting Source \( i \) with probability \( \mu^{ns}_i \), i.e. $u_i(t)=1$, then the waiting time is given by
    \begin{equation}
    W_i = 
    \begin{cases}
        1, & \text{with probability } p_i, \\
        1 + W_i, & \text{with probability } 1 - p_i,
    \end{cases}
    \end{equation}

    \item Selecting Source \( j \neq i \) with probability \( \mu^{ns}_j \), i.e. $u_j(t)=1$, then the waiting time is given by
            \begin{equation}
    W_i = Y_j + W_i
    \end{equation}
        \begin{equation}
        Y_j = 
        \begin{cases}
            1, & \text{with probability } 1 - p_j, \\
            1 + S_j, & \text{with probability } p_j,
        \end{cases}
        \end{equation}
        where each \( S_j \) is service time of source $j$, and the total time spent include and after selecting source \( j \) is \( Y_j \), after which the BS returns to $Selecting$.
\end{itemize}

Next, we calculate the first moment \( \mathbb{E}[W_i] \) and the second moment \( \mathbb{E}[W^2_i] \) of the waiting time, respectively. We split the first moment \( \mathbb{E}[W_i] \)  conditional on the scheduling decision as
\begin{equation}\label{E[W^ns|condition]}
\mathbb{E}[W_i] = \mu^{ns}_i \mathbb{E}[W_i \mid u_i(t)=1] + \sum_{j\neq i} \mu^{ns}_j \mathbb{E}[W_i \mid u_j(t)=1].
\end{equation}
where the conditional expectation are given by
\begin{equation}\label{E[W|i]}
\mathbb{E}[W_i \mid u_i(t)=1] = p_i \cdot 1 + (1-p_i)(1 + \mathbb{E}[W_i]) = 1 + (1-p_i)\mathbb{E}[W_i].
\end{equation}
\begin{equation}\label{E[W|j]}
\mathbb{E}[W_i \mid u_j(t)=1] = \mathbb{E}[Y_j] + \mathbb{E}[W_i].
\end{equation}

Taking the expectation of total time spent after selecting source \( Y_ji \) and substituting with~\eqref{E[S^ns]} yields
\begin{equation}\label{E[Y]}
\mathbb{E}[Y_j] = (1-p_j)\cdot 1 + p_j\left(1 + \mathbb{E}[S_j]\right) = L_j.
\end{equation}

Substituting~\eqref{E[Y]} into~\eqref{E[W|j]} yields
\begin{equation}\label{E[W_i^ns|j]}
\mathbb{E}[W_i \mid u_j(t)=1] = L_i + \mathbb{E}[W_i].
\end{equation}

Thus, substituting~\eqref{E[W|i]} and ~\eqref{E[W_i^ns|j]} into~\eqref{E[W^ns|condition]} yields 
\begin{equation}\label{E[W^ns]}
\mathbb{E}[W_i] = \frac{\mu^{ns}_i + \sum_{j \neq i}\mu^{ns}_j L_j}{\mu^{ns}_i p_i}.
\end{equation}

Similarly, we condition on the scheduling decision to write
\begin{equation}\label{E[W^2^ns|condition]}
E[W_i^2] = \mu^{ns}_i\,\mathbb{E}[W_i^2 \mid u_i(t)=1] + \sum_{j\neq i}\mu^{ns}_j\,\mathbb{E}[W_i^2 \mid u_j(t)=1].
\end{equation}

And the expectation conditional on source $i$ is selected is given by
\begin{equation}
   \mathbb{E}[W_i^2 \mid u_i(t)=1] = p_i\cdot 1^2 + (1-p_i)\mathbb{E}[(1+W_i)^2],
\end{equation}
Expanding the square yields
\begin{equation}\label{E[W^2^{ns}|i]} 
\mathbb{E}[W_i^2 \mid u_i(t)=1] = p_i + (1-p_i)\Bigl(1 + 2\mathbb{E}[W_i] + \mathbb{E}[W_i^2]\Bigr) = 1 + 2(1-p_i)\mathbb{E}[W_i] + (1-p_i)\mathbb{E}[W_i^2].
\end{equation} 
Here, the waiting time for source \(i\) is the sum of the time \(Y_j\) spent after selecting source \(j\) and the subsequent waiting time. And the expectation conditional on source $j\neq i$ is selected is given by
\begin{equation}\label{E[W^2^{ns}|j]}
E[W_i^2 \mid u_j(t)=1] = E\Bigl[(Y_j+W_i)^2\Bigr] = \mathbb{E}[Y_j^2] + 2\mathbb{E}[Y_j]\mathbb{E}[W_i] + \mathbb{E}[W_i^2].
\end{equation}

Substituting~\eqref{E[W^2^{ns}|i]} and~\eqref{E[W^2^{ns}|j]}into~\eqref{E[W^2^ns|condition]} gives
\begin{align}
\mathbb{E}[W_i^2] &= \mu^{ns}_i\Bigl[1+2(1-p_i)\mathbb{E}[W_i] + (1-p_i)\mathbb{E}[W_i^2]\Bigr] \nonumber\\
&+ \sum_{j\neq i}\mu^{ns}_j\Bigl[\mathbb{E}[Y_j^2] + 2\mathbb{E}[Y_j]\mathbb{E}[W_i] + \mathbb{E}[W_i^2]\Bigr].
\end{align}

Collecting the terms involving \(\mathbb{E}[W_i^2]\) leads to
\begin{equation}\label{E[W^2_final]}
\mathbb{E}[W_i^2] = \frac{\mu^{ns}_i\Bigl[1+2(1-p_i)\mathbb{E}[W_i]\Bigr] + \sum_{j\neq i}\mu^{ns}_j\Bigl[\mathbb{E}[Y_j^2] + 2\mathbb{E}[Y_j]\mathbb{E}[W_i]\Bigr]}{1- \sum_{j}\mu^{ns}_j +p_i\mu^{ns}_i }.
\end{equation}

Where the $\mathbb{E}[Y_j^2]$ is given by 
\begin{equation}\label{E[Y^2=S]}
\mathbb{E}[Y_i^2] = (1-p_i)\cdot1^2 + p_i \mathbb{E}[(1+S_i)^2].
\end{equation}

Substituting~\eqref{E[S^ns]} and~\eqref{E[S^ns^2]} into~\eqref{E[Y^2=S]} yields
\begin{equation}\label{E[Y^2]}
\mathbb{E}[Y_i^2] = 1 + 2(L_i -1) + \frac{(L_i -1)(L_i - p_i)}{p_i}.
\end{equation}

Thus, we obtain a close-from expression of NSRPs by organizing~\eqref{E[j^ns]=E[W.S.]},~\eqref{E[S^ns]},~\eqref{E[S^ns^2]},~\eqref{E[W^ns]},~\eqref{E[W^2_final]}, and~\eqref{E[Y^2]}.


\begin{equation}
\begin{aligned}
  \mathbb{E}[J^{ns}] 
  = \frac{1}{N} \sum_{i=1}^N \alpha_i \Biggl[\frac{\mu^{ns}_i p_i}{\sum^N_{j=1}\mu^{ns}_j L_j}
    \Biggl(\frac{\mathbb{E}[W^2_i] + \mathbb{E}[S^2_i]}{2} 
   + \left(\frac{L_i-1}{p_i}\right)^2
+ 2\mathbb{E}[W^2_i]
      \left(\frac{L_i-1}{p_i}\right)\Biggl) +1
  \Biggr]. 
\end{aligned}
\end{equation}
Where
\begin{align}
    \mathbb{E}[S^2_i] &= \frac{(L_i - 1)(L_i - p_i)}{p_i^2}, \\
    \mathbb{E}[W_i]&=\frac{\sum^N_{j=1}\mu^{ns}_j L_j- \mu^{ns}_i(L_i-1)}{\mu^{ns}_i p_i}.\\
    \mathbb{E}[W_i^2] &=  \frac{\mu^{ns}_i\Bigl[1+2(1-p_i)\mathbb{E}[W_i]\Bigr] + \sum_{j\neq i}\mu^{ns}_j\Bigl[\mathbb{E}[Y_j^2] + 2\mathbb{E}[Y_j]\mathbb{E}[W_i]\Bigr]}{1- \sum_{j}\mu^{ns}_j +p_i\mu^{ns}_i }., \\
   \mathbb{E}[Y_i^2]&= 2L_i -1 + \frac{(L_i -1)(L_i - p_i)}{p_i}.
\end{align}

\end{proof}

\section{Upper Bound for Lyapunov Drift}
\label{Upper Bound for Lyapunov Drift}
In this appendix, we obtain the expressions in~\eqref{lyapunov drift upper bound1}--~\eqref{lyapunov drift upper bound3}, which represent an upper bound on the Lyapunov Drift. Consider the network state $\mathbb{S}(t) := \{h_i(t),z_i(t),L_i(t), x_i (t)\}_{i=1}^N$, the Lyapunov Function  in~\eqref{lyapunov function} and the Lyapunov Drift $\Delta \left(\mathbb{S}(t)\right)$  in~\eqref{lyapunov Drift}. Substituting~\eqref{lyapunov function} into~\eqref{lyapunov Drift}, we get
\begin{equation}
\begin{aligned}\label{Lyapunov Drift Detail}
        \Delta&\left(\mathbb{S}(t)\right)=\frac{V}{2}\mathbb{E}\left[\sum_{i=1}^{N}\left[x^+_{i}(t+1)\right]^2-\sum_{i=1}^{N}\left[x^+_{i}(t)\right]^2|\mathbb{S}(t)\right]\\&+\mathbb{E}\left[\sum_{i=1}^{N}\beta_{i}\left[h_{i}(t+1)-z_{i}(t+1)\right]^2-\sum_{i=1}^{N}\beta_{i}\left[h_{i}(t)-z_{i}(t)\right]^2|\mathbb{S}(t)\right]\\&+\mathbb{E}\left[\sum_{i=1}^{N}\gamma_{i}\left[z_{i}(t+1)+L_{i}(t+1)\right]^2-\sum_{i=1}^{N}\gamma_{i}\left[z_{i}(t)+L_{i}(t)\right]^2|\mathbb{S}(t)\right]
\end{aligned}
\end{equation}
Recall that the evolution of $L_i(t)$, $z_i(t)$, and $h_i(t)$ are given by ~\eqref{Levolves}, ~\eqref{zevolves} and~\eqref{hevolves}, respectively. The evolution of $h_i(t)-z_i(t)$ and 
$z_i(t)+L_i(t)$ are given by
\begin{equation}\label{h-z evolution}
   h_i(t+1)-z_i(t+1)= \begin{cases}
            h_i(t)-z_i(t)+1 & \text{if } d_i(t)=0 \text{ and } L_i(t)=L,\\ 
            z_i(t) & \text{if } d_i(t)=1\text{ and }L_i(t)=1, \\ 
h_i(t)-z_i(t) &\text{otherwise.}
    \end{cases}
\end{equation}

\begin{equation}\label{z+L evolution}
z_i(t+1) + L_i(t+1) = 
\begin{cases}
z_i(t) + L_i(t),
& \text{if } d_i(t)=1 \text{ and }L_i(t)>1 \\
& \text{ or } L_i(t)=L,\\
L_i + 1, 
& \text{if } d_i(t)=1\text{ and }L_i(t)=1,\\
z_i(t) + L_i(t) + 1, 
& \text{otherwise.}
\end{cases}
\end{equation}
Substitute~\eqref{h-z evolution} and~\eqref{z+L evolution} into~\eqref{Lyapunov Drift Detail}, leverage the discuss of throughput debt in~\cite[Appendix A]{kadota2019scheduling}, and rearrange the terms, we obtain ~\eqref{lyapunov drift upper bound1}--\eqref{lyapunov drift upper bound3}.

\section{Proof of Theorem~\ref{theo: Performance Bound for MW}}\label{Proof of Theorem theo: Performance Bound for MW}



The expression for the Lyapunov drift~\eqref{lyapunov drift upper bound1}--\eqref{lyapunov drift upper bound3} is central to the analysis in this appendix and is rewritten below for convenience.

\begin{equation*}
        \Delta\left(\mathbb{S}(t)\right)\leq B(t)-\sum_{i=1}^{N}p_{i}\mathbb{E}\left[u_{i}(t)|\mathbb{S}(t)\right]C_{i}(t)
\end{equation*}

where
\begin{equation*}
\begin{aligned}
B(t)=&\sum_{i=1}^{N}\beta_{i}\mathbb{I}_{L_i(t)=L}\left[2h_{i}(t)-1\right]+V\left[x^+_i(t)\bar{q}_i+\frac{1}{2}\right] +\sum_{i=1}^{N}\gamma_{i}\mathbb{I}_{1<L_i(t)<L_i}\left[2\left(z_{i}(t)+L_{i}(t)\right)-1\right],
\end{aligned}
\end{equation*}
\begin{equation*}
\begin{aligned}
C_{i}(t) & =\beta_{i}\mathbb{I}_{L_i(t)=L}\left[2h_{i}(t)-1\right]+\beta_{i}\mathbb{I}_{L_i(t)=1}\left[h_{i}^{2}(t)-2h_{i}(t)z_{i}(t)\right]\\
 & +\gamma_{i}\mathbb{I}_{1<L_i(t)<L_i}\left[2z_{i}(t)+2L_{i}(t)-1\right]+\gamma_{i}\mathbb{I}_{L_i(t)=1}\left[\left(z_{i}(t)+2\right)^{2}-\left(L_{i}+1\right)^{2}\right]+ Vx^+(t).
\end{aligned}
\end{equation*}

\vspace{0.5ex}\noindent\textbf{Step 1: A throughput lemma for Max-Weight.}
Prior to deriving the upper bound on the EWSAoI achieved by the Max-Weight policy, we prove Lemma~\ref{mw lemma}. 

\begin{lemma}\label{mw lemma}
    If there exists an SRP $s$ with long-term packet throughput greater than the throughput target for all sources~$i$, i.e., $p_i\mu^s_i\geq \bar{q_i}, \forall i$, then the long-term packet throughput of the Max-Weight policy is also such that $q_i^{MW}\geq \bar{q_i}, \forall i$.
\end{lemma}
\begin{proof}
We compare MW with an arbitrary stationary randomized policy (SRP)~$s$ and show that MW yields a no-larger drift term. We then substitute this bound into the
drift expression, take expectations, sum over time, and let the horizon
$T\to\infty$ to prove that the throughput debt process remains finite. By
applying the stability result in~\cite[Theorem~2.8]{neely2010stability}, we
conclude that the throughput targets are met, which establishes
Lemma~\ref{mw lemma}. The proof relies on Lyapunov-drift techniques similar to
those in~\cite[Appendix~B]{kadota2019scheduling}, and the full derivation is
provided in Appendix~\ref{Proof of Lemma mw}.
\end{proof}

Next, we leverage Lemma~\ref{mw lemma} to derive the upper bound on the EWSAoI associated with the Max-Weight policy described in Theorem~\ref{theo: Performance Bound for MW}. 

Consider the throughput targets below 
\begin{equation}
    \bar{q}_i = q_i^{L_B}-\sigma, \quad \forall i,
\end{equation}
with arbitrarily small $\sigma>0$. Clearly, $\{\bar{q}_i\}_{i=1}^N$ is feasible since a SRP with $\mu_i^{L_B}=q_i^{L_B}/p_i$ can meet the throughput target. By setting $\{\bar{q}_i\}_{i=1}^N$ as the target value in the throughput debt $x^+_i(t)$, leveraging Lemma~\ref{mw lemma}, the throughput achieved by Max-Weight policy, i.e. $\{q^{MW}_i\}_{i=1}^N$, satisfies 
\begin{equation}\label{q^{MW}_i > q_i^{L_B}}
    q^{MW}_i \geq q_i^{L_B}-\sigma, \quad \forall i.
\end{equation}
\vspace{0.5ex}\noindent\textbf{Step 2: Drift inequality under the Max-Weight policy.}
Summing~\eqref{lyapunov drift upper bound1} over $t \in
\{1, 2,\ldots, T\}$, taking expectation with respect to $\mathbb{S}(t)$, taking the limit as $T\rightarrow\infty$, dividing by $TN$, and then using $q^{MW}_i=p_i\mathbb{E}\left[u_i(t)\right]$, we obtain
$LHS'_1+LHS'_2+LHS'_3\leq RHS'_1 + RHS'_2 + RHS'_3$, where
\begin{equation}\label{LHS'1}
\begin{aligned}
LHS'_1=&\lim_{T\rightarrow\infty}\frac{1}{TN}\sum_{t=1}^{T}\sum_{i=1}^{N}q^{MW}_i\beta_{i}\mathbb{E}\left[\mathbb{I}_{L_i(t)=1}|u_i(t)=1\right]\mathbb{E}\left[h^2_{i}(t)-2h_{i}(t)z_{i}(t)|u_i(t)=1,L_i(t)=1\right]
\end{aligned}
\end{equation}
\begin{equation}\label{LHS'2}
\begin{aligned}
LHS'_2=&\lim_{T\rightarrow\infty}\frac{1}{TN}\sum_{t=1}^{T}\sum_{i=1}^{N}q^{MW}_i\gamma_{i}\mathbb{E}\left[\mathbb{I}_{L_i(t)=1}|u_i(t)=1\right]\mathbb{E}\left[\left(z_{i}(t)+2\right)^{2}-\left(L_{i}+1\right)^{2}|u_i(t)=1\!,L_i(t)=1\!\right]
\end{aligned}
\end{equation}
\begin{equation}\label{LHS'3}
\begin{aligned}
LHS'_3=&    \lim_{T\rightarrow\infty}\frac{1}{TN}\sum_{t=1}^{T}\sum_{i=1}^{N}V\left(q^{MW}_i-\bar{q}_{i}\right)\mathbb{E}[x_{i}^{+}(t)|u_i(t)=1,\!L_i(t)=1]
\end{aligned}
\end{equation}

\begin{equation}\label{RHS'1}
\begin{aligned}
RHS'_1=&\lim_{T\rightarrow\infty}\frac{1}{TN}\sum_{t=1}^{T}\sum_{i=1}^{N}\beta_{i}\mathbb{E}[\mathbb{I}_{L_i(t)=L}\left(2h_{i}(t)-1\right)]
\end{aligned}
\end{equation}
\begin{equation}\label{RHS'2}
\begin{aligned}
RHS'_2\!=&\!\lim_{T\!\rightarrow\infty}\!\frac{1}{TN}\!\sum_{t=1}^{T}\!\sum_{i=1}^{N}\!\gamma_{i}\mathbb{E}[\mathbb{I}_{1<L_i(t)<L_i}\!\left(2z_{i}(t)\!+\!2L_{i}(t)\!-\!1\right)]
\end{aligned}
\end{equation}
\begin{equation}\label{RHS'3}
RHS'_3=   \frac{NV}{2}
\end{equation}

Next, we analyze the packet transmission process and develop bounds for each term~\eqref{LHS'1}--\eqref{RHS'3}.
Notice that a necessary and sufficient condition for each update delivery of source $i$ is given by the indicator $d_{i}(t)\mathbb{I}_{L_i(t)=1}=1$, and that the time average update delivery rate of source $i$ is $q_{i}/L_{i}$. Then, assuming ergodicity of the packet delivery process, we can establish 
\begin{equation}\label{subs1}
    \frac{q_{i}}{L_{i}}=\mathbb{E}[d_{i}(t)\mathbb{I}_{L_i(t)=1}] \; ,
\end{equation}
expanding the expectation on the RHS yields
\begin{equation}
\begin{aligned}\label{subs2}
        \mathbb{E}[d_{i}(t)\mathbb{I}_{L_i(t)=1}]&=p_{i}\mathbb{E}\left[u_{i}(t)\mathbb{E}[\mathbb{I}_{L_i(t)=1}|u_{i}(t)]\right]\\&=p_{i}\mathbb{P}\left[u_{i}(t)=1\right]\mathbb{E}[\mathbb{I}_{L_i(t)=1}|u_{i}(t)=1]
\end{aligned}
\end{equation}
Recall that $q_{i}=p_{i}\mathbb{P}\left[u_{i}(t)=1\right]$, further substituting~\eqref{subs1} into~\eqref{subs2} yields
\begin{equation}\label{throughput when selected}
    \mathbb{E}[\mathbb{I}_{L_i(t)=1}|u_{i}(t)=1]=\frac{1}{L_{i}}
\end{equation}

    Since the Max-Weight policy selects the source with the highest value of $C_i(t)$ at any given network state $\mathbb{S}(t)$, we establish that
\begin{equation}\label{cancel u_i(t=1)}
    \mathbb{E}\left[C_i(t)|u_i(t)=1\right] \geq \mathbb{E}\left[C_i(t)\right]
\end{equation}

\vspace{0.5ex}\noindent\textbf{Step 3: A refined ordering property for $C_i(t)$.}
Next we show that, for a low enough value of $V$, the value of $C_i(t)$ in
\eqref{lyapunov drift upper bound3} when the last packet of an update is successfully transmitted 
is no smaller than the value of $C_i(t)$ when earlier packets of the same update were successfully transmitted.

\begin{lemma}\label{lem:Ci-ordering}
For a given source $i$, a given update index $m$, and for a sufficiently small $V>0$, 
the value of $C_i(t)$ in~\eqref{lyapunov drift upper bound3} when the last packet of update $m$ is successfully transmitted 
is no smaller than the value of $C_i(t)$ when earlier packets of update $m$ were successfully transmitted, i.e., 
\begin{equation}\label{C_i(t) ordering}
    C_i\bigl(t_i[m]\bigr)\;\ge\;C_i(t),
\forall t'[m] \leq t \leq t_i[m],\text{ s.t. }u_i(t)c_i(t)=1.
\end{equation}
\end{lemma}

\begin{proof}
We compare the value $C_i(t[m])$ with that at the first successful packet $C_i(t'[m])$, as well as with its values at the intermediate packet-delivery slots. By performing a simple case split, bounding the resulting differences through algebraic manipulation, and choosing $V>0$ sufficiently small, we show that the last successful packet maximizes $C_i(t)$ over all successful packets of update $m$, thereby obtaining~\eqref{C_i(t) ordering}. Detailed derivations are provided in Appendix~\ref{Proof of lemma Ci-ordering}.
\end{proof}


\vspace{0.5ex}\noindent\textbf{Step 4: An inequality on the conditional expectation of $C_i(t)$.}
We now use Lemma~\ref{lem:Ci-ordering} to derive an inequality for the 
conditional expectations of $C_i(t)$, which will be instrumental for the 
subsequent drift analysis.
\begin{lemma}\label{lem:C-last-avg}
For each source $i$, the conditional expected value of $C_i(t)$ at the
slots where the \emph{last} packet of an update is delivered (i.e.,
$L_i(t)=1$ and $u_i(t)=1$) is no smaller than the conditional expected
value of $C_i(t)$ over all slots in which source $i$ is scheduled (i.e.,
$u_i(t)=1$), namely
\begin{equation}\label{C ineq for infty T}
    \mathbb{E}\bigl[C_{i}(t)\mid L_{i}(t)=1,u_{i}(t)=1\bigr]
    \;\ge\;
    \mathbb{E}\bigl[C_{i}(t)\mid u_{i}(t)=1\bigr].
\end{equation}
\end{lemma}

\begin{proof}
Using Lemma~\ref{lem:Ci-ordering} and comparing the value of $C_i(t[m])$ with the average of $C_i(t)$ over all packet-delivery slots of the same update yields a per-update inequality. Summing this inequality over all updates, normalizing over time, and letting $T\to\infty$, assuming the ergodicity\footnote{
Under the Max–Weight scheduling policy, the joint network state
$\mathbb{S}(t)\triangleq\{h_i(t),\,z_i(t),\,L_i(t)\}_{i=1}^{N}$ forms a
discrete-time Markov chain that is \emph{irreducible} and \emph{aperiodic}.
If the Max-Weight policy stabilizes the network, then $h_i(t)$, $z_i(t)$, and
$L_i(t)$ remain bounded almost surely, thus the chain is \emph{positive
recurrent}.  Irreducibility, aperiodicity, and positive recurrence together
ensure the existence of a unique stationary distribution, and hence
$\mathbb{S}(t)$ is ergodic.} of $\mathbb{S}(t)$ turns time averages into steady-state expectations. Finally, using the relation $L_i\,\mathbb{E}[\mathbb{I}_{L_i(t)=1}\mathbb{I}_{u_i(t)=1}]  = \mathbb{E}[\mathbb{I}_{u_i(t)=1}]$, we obtain~\eqref{C ineq for infty T}. Detailed derivations are given in Appendix~\ref{Proof of lemma C-last-avg}.
\end{proof}

\vspace{0.5ex}\noindent\textbf{Step 5: Bounding the drift and concluding the EWSAoI bound.}
Substitute~\eqref{throughput when selected},~\eqref{cancel u_i(t=1)},
and ~\eqref{C ineq for infty T} into~\eqref{LHS'1} and~\eqref{LHS'2}, we obtain
\begin{equation}\label{LHS''1}
\begin{aligned}
& LHS'_1 + LHS'_2 + LHS'_3 \\
& \stackrel{(a)}{\geq}
  \lim_{T\to\infty}\!\frac{1}{T}\!\sum_{t=1}^{T}\sum_{i=1}^{N}
  \frac{q^{MW}_i\beta_i}{L_i}\,
  \mathbb{E}\!\left[\,h_i^{2}(t)\!-\!2h_i(t)z_i(t)\right] \!+ \!\!
  \lim_{T\to\infty}\!\frac{1}{T}\!\sum_{t=1}^{T}\sum_{i=1}^{N}
  \frac{q^{MW}_i\gamma_i}{L_i}\!
  \mathbb{E}\!\left[(z_i(t)+2)^{2}\!\!-\!(L_i+1)^{2}\!\right] \\
& \stackrel{(b)}{\geq}
  \lim_{T\to\infty}\frac{1}{T}\sum_{t=1}^{T}\sum_{i=1}^{N}
  \frac{q^{MW}_i\beta_i}{L_i}\,
  \mathbb{E}\!\left[\,(h_i(t)-2z_i(t))^{2}\right] \!+\!\! 
  \lim_{T\to\infty}\frac{1}{T}\sum_{t=1}^{T}\sum_{i=1}^{N}
  \frac{q^{MW}_i\gamma_i}{L_i}\!
  \mathbb{E}\!\left[\!(z_i(t)+2)^{2}\!\!-\!(L_i+1)^{2}\!\right] \\
& \stackrel{(c)}{\geq}\!
  \lim_{T\to\infty}\!\frac{1}{T}\!\sum_{t=1}^{T}\sum_{i=1}^{N}
  \frac{q^{MW}_i\beta_i}{L_i}\!\Bigl[\mathbb{E}\!\left[h_i(t)\right]
  \!-\! 2\mathbb{E}\!\left[z_i(t)\right]\Bigr]^{2}  \!+\!\! 
  \lim_{T\to\infty}\!\frac{1}{T}\!\sum_{t=1}^{T}\sum_{i=1}^{N}
  \frac{q^{MW}_i\gamma_i}{L_i}
  \Bigl(\mathbb{E}\!\left[z_i(t)+2\right]^{2}
  \!\!-\! (L_i+1)^{2}\!\Bigr) \\
& \stackrel{(d)}{\geq}
  \lim_{T\to\infty}\frac{1}{T}\sum_{t=1}^{T}\sum_{i=1}^{N}
  \frac{q^{MW}_i\beta_i}{L_i}\,
  \Bigl[\mathbb{E}\!\left[\,h_i(t)\right]
  - 2\,\frac{L_i}{\bar{q}_i}\Bigr]^{2}+\lim_{T\to\infty}\!\frac{1}{T}\!\sum_{t=1}^{T}\sum_{i=1}^{N}
  \frac{q^{MW}_i\gamma_i}{L_i}
  \Bigl(\frac{-L_i^2+8L_i+24}{2}\!\Bigr)
\end{aligned}
\end{equation}
where  $(a)$ from~\eqref{throughput when selected},~\eqref{cancel u_i(t=1)}
,~\eqref{C ineq for infty T} and $LHS'_3>0$ since where $x^+_i(t)\geq 0 $ and by~\eqref{q^{MW}_i > q_i^{L_B}}; 
$(b)$ from $h_i(t)>h_i(t)-2z_i(t)$;
$(c)$ by applying Jensen Inequality $\mathbb{E}\left[\left(\bar{h}_{i}(t)-2z_{i}(t)\right)^2\right]\geq\left[\mathbb{E}\left[h_{i}(t)\right]-2\mathbb{E}\left[z_{i}(t)\right]\right]^2$, where $\bar{h}_i(t)=\max(h_i(t),2z_i(t))$ and $\mathbb{E}\left[\bar{h}_{i}(t)\right]\geq\mathbb{E}\left[h_{i}(t)\right]$;
and $(d)$ from
\begin{equation}
    \frac{L_i+1}{2}\leq\mathbb{E}\left[z_{i}(t)\right]\leq\frac{ L_i}{\bar{q}_i}
\end{equation}

Since $\mathbb{I}_{L_i(t)=L}\leq1,\mathbb{I}_{1<L_i(t)<L_i}\leq1$, we obtain
\begin{equation}\label{RHS''1}
RHS_1'\leq\lim_{T\rightarrow\infty}\frac{1}{TN}\sum_{t=1}^{T}\sum_{i=1}^{N}\beta_i
\mathbb{E}\bigl[2h_{i}(t)-1\bigr].
 \end{equation}
\begin{equation}\label{RHS''2}
RHS'_2\leq\lim_{T\rightarrow\infty}\frac{1}{TN}\sum_{t=1}^{T}\sum_{i=1}^{N}\gamma_{i}\mathbb{E}[2z_{i}(t)+2L_{i}(t)-1]
\end{equation}
Substituting~\eqref{LHS''1}, ~\eqref{RHS''1}, and~\eqref{RHS''2} into $LHS'_1+LHS'_2+LHS'_3\leq RHS'_1 + RHS'_2 + RHS'_3$ yields
\begin{equation}\label{simplified ineq}
\begin{aligned}
        &\lim_{T\to\infty}\frac{1}{T}\sum_{t=1}^{T}\sum_{i=1}^{N}
  \frac{q^{MW}_i\beta_i}{L_i}\,
  \Bigl[\mathbb{E}\!\left[\,h_i(t)\right]
  - 2\,\frac{L_i}{\bar{q}_i}\Bigr]^{2}+\sum_{i=1}^{N}
  \frac{q^{MW}_i\gamma_i}{L_i}
  \Bigl(\frac{-L_i^2+8L_i+24}{2}\!\Bigr)\\ &\leq \lim_{T\to\infty}\frac{1}{T}\sum_{t=1}^{T}\sum_{i=1}^{N}\beta_{i}\mathbb{E}\left[2h_{i}(t)-1\right]+\frac{NV}{2}+\lim_{T\to\infty}\frac{1}{T}\sum_{t=1}^{T}\sum_{i=1}^{N}\gamma_{i}\mathbb{E}\left[2z_{i}(t)+2L_{i}(t)-1\right]
\end{aligned}
\end{equation}

Substituting $\beta_i=\alpha_iL_i/\bar{q}_i,\gamma_i=\alpha_iL_i/\bar{q}_i\sqrt{p_i}$ and considering $\sigma\rightarrow0$ yields
\begin{equation}\label{substituted simplified ineq}
\begin{aligned}
        &\lim_{T\to\infty}\frac{1}{T}\sum_{t=1}^{T}\sum_{i=1}^{N}
  \alpha_i\,
  \Bigl[\mathbb{E}\!\left[\,h_i(t)\right]
  - 2\,\frac{L_i}{\bar{q}_i}\Bigr]^{2}+\sum_{i=1}^{N}
  \frac{q^{MW}_i\gamma_i}{L_i}
  \Bigl(\frac{-L_i^2+8L_i+24}{2}\!\Bigr)\\ &\leq \lim_{T\to\infty}\frac{1}{T}\sum_{t=1}^{T}\sum_{i=1}^{N}\frac{\alpha_iL_i}{\bar{q}_i}\mathbb{E}\left[2h_{i}(t)-1\right]+\frac{NV}{2}+\lim_{T\to\infty}\frac{1}{T}\sum_{t=1}^{T}\sum_{i=1}^{N}\frac{\alpha_iL_i}{\bar{q}_i\sqrt{p_i}}\mathbb{E}\left[2z_{i}(t)+2L_{i}(t)-1\right]
\end{aligned}
\end{equation}


Consider the Cauchy-Schwarz inequality
\begin{equation} \label{theorem 7 cathy}
\begin{aligned}
\left(\sum_{i=1}^{N}\alpha_i\left[\mathbb{E}\left[h_{i}(t)\right]-\frac{3L_{i}}{\bar{q}_i}\right]^{2}\right)\left(\sum_{i=1}^{N}\alpha_i\right)\quad\geq\left|\sum_{i=1}^{N}\alpha_i\left[\mathbb{E}\left[h_{i}(t)\right]-\frac{3L_{i}}{\bar{q}_i}\right]\right|^{2}
\end{aligned}
\end{equation}
Applying~\eqref{theorem 7 cathy} and $\mathbb{E}\left[z_{i}(t)\right]\leq L_i/\bar{q}_i$ into~\eqref{substituted simplified ineq} yields
\begin{equation} \label{theorem 7 square}
\begin{aligned}
&\left(\sum_{i=1}^{N}\alpha_i\left[\mathbb{E}\left[h_{i}(t)\right]-\frac{3L_{i}}{\bar{q}_i}\right]^{2}\right)\!\left(\sum_{i=1}^{N}\alpha_i\!\right)\!\leq\!\left(\!\sum_{i=1}^{N}\alpha_i\!\right)\!\Biggl(\sum_{i=1}^{N}\frac{\alpha_i}{\sqrt{p_i}}\Bigl[5\frac{L^2_{i}\sqrt{p_i}}{\bar{q}^2_i}\!-\!\frac{L_{i}}{\bar{q}_{i}}\!-\!\frac{L_{i}\sqrt{p_i}}{\bar{q}_{i}}\!+\!2\frac{L^2_i}{\bar{q}^2_i}\!+\!2\frac{L_i^2}{\bar{q}_i}\!+\!\frac{-L_i^2+8L_i+24}{2}\Bigl]\Biggl)
\end{aligned}
\end{equation}
Thus, we obtain the upper bound of achieved by Max-Weight policy, which is given by
\begin{equation}\label{EJMW conclusion}
\mathbb{E}\left[J^{MW}\right]\leq\frac{1}{N}\sum_{i=1}^{N}\frac{3\alpha_iL_{i}}{\bar{q}_i}+\frac{1}{N}\sqrt{\varPsi}
\end{equation}
where
\begin{equation}
\begin{aligned}
&\varPsi=\left(\sum_{i=1}^{N}\alpha_i\right)\Biggl(\sum_{i=1}^{N}\frac{\alpha_i}{\sqrt{p_i}}\Bigl[5\frac{L^2_{i}\sqrt{p_i}}{\bar{q}^2_i}-\frac{L_{i}}{\bar{q}_{i}}-\frac{L_{i}\sqrt{p_i}}{\bar{q}_{i}}+2\frac{L^2_i}{\bar{q}^2_i}+2\frac{L_i^2}{\bar{q}_i}+\frac{-L_i^2+8L_i+24}{2}\Bigl]\Biggl)
\end{aligned}
\end{equation}

Comparing~\eqref{EJMW conclusion} with~\eqref{performancelowerbound} yields
\begin{equation}
\rho=6+\frac{\sqrt{\varPsi}}{NL_B}
\end{equation}
where $\rho$ is the optimal ratio of Max-Weight policy.

\section{Proof of Lemma~\ref{mw lemma}}
\label{Proof of Lemma mw}
Recall that the Max-Weight policy minimizes the RHS of~\eqref{lyapunov drift upper bound1} by selecting $i = arg \max{p_iC_i(t)}$ in every slot $t$.
Hence, any other policy $\pi \in \Pi$ yields a higher (or equal) RHS. Consider a  SRP $s\in \Pi^s_r$ that, in each slot $t$, selects node $i$ with probability $\mu^s_i \in (0, 1]$. The scheduling decision of policy $s$ is independent of the network state $\mathbb{S}(t)$, and thus
\begin{equation}\label{muupperboundMW}
\begin{aligned}
\sum_{i=1}^N p_i\mathbb{E}\left[u_i(t)|\mathbb{S}(t)\right]C_i(t)\geq  \sum_{i=1}^N p_i\mu_i^sC_i(t) 
\end{aligned}
\end{equation}
where $u^S_i(t)$ denote the scheduling decision made by optimal SRP $S$.
Substituting~\eqref{muupperboundMW} into the equation of the Lyapunov Drift gives
\begin{equation}\label{substitued expected drift}
\begin{gathered} \Delta\left(\mathbb{S}(t)\right)+\sum_{i=1}^{N}p_{i}\mu^s_{i}\mathbb{I}_{L_i(t)=1}\left[\beta_{i}\left(h^2_{i}(t)-2h_{i}(t)z_{i}(t)\right)+\gamma_{i}\left(\left(z_{i}(t)+2\right)^{2}-\left(L_{i}+1\right)^{2}\right)\right]+\sum_{i=1}^{N}V\left(p_{i}\mu^s_{i}-q_{i}\right)x_{i}^{+}(t)\\
\leq\sum_{i=1}^{N}\beta_{i}\mathbb{I}_{L_i(t)=L}\left[2h_{i}(t)-1\right]+\sum_{i=1}^{N}\gamma_{i}\mathbb{I}_{1<L_i(t)<L_i}\left[2z_{i}(t)+2L_{i}(t)-1\right]+\frac{NV}{2}
\end{gathered}
\end{equation}

For simplicity of exposition, we divide inequality~\eqref{substitued expected drift} into six terms $\Delta\left(\mathbb{S}(t)\right)+LHS_1+LHS_2+LHS_3\leq RHS_1 + RHS_2 + RHS_3$. Taking expectation of~\eqref{substitued expected drift} with respect to $\mathbb{S}(t)$, summing over $t \in
\{1, 2,\ldots, T\}$ taking the limit as $T\rightarrow\infty$, and then dividing by $TN$, gives

\begin{equation}
    \lim_{T\rightarrow\infty}\frac{1}{TN}\sum_{t=1}^{T}\mathbb{E}\left[\Delta\left(\mathbb{S}(t)\right)\right]=\frac{\mathbb{E}\left[\mathcal{L}(T)\right]}{TN}=0
\end{equation}
\begin{equation}\label{LHS1}\begin{aligned}LHS_1=&\lim_{T\rightarrow\infty}\frac{1}{TN}\sum_{t=1}^{T}\sum_{i=1}^{N}p_{i}\mu^s_{i}\beta_{i}\mathbb{E}[\mathbb{I}_{L_i(t)=1}]\mathbb{E}\left[h^2_{i}(t)-2h_{i}(t)z_{i}(t) | L_i(t)=1\right]
\end{aligned}
\end{equation}
\begin{equation}\label{LHS2}\begin{aligned}
LHS_2=&\lim_{T\rightarrow\infty}\frac{1}{TN}\sum_{t=1}^{T}\sum_{i=1}^{N}p_{i}\mu^s_{i}\gamma_{i}\mathbb{E}[\mathbb{I}_{L_i(t)=1}]\mathbb{E}\left[\left(z_{i}(t)+2\right)^{2}-\left(L_{i}+1\right)^{2}| L_i(t)=1\right]\end{aligned}
\end{equation}
\begin{equation}\label{LHS3}
LHS_3=    \lim_{T\rightarrow\infty}\frac{1}{TN}\sum_{t=1}^{T}\sum_{i=1}^{N}V\left(p_{i}\mu^s_{i}-q_{i}\right)\mathbb{E}[x_{i}^{+}(t)]
\end{equation}
\begin{equation}\label{RHS1}
\begin{aligned}
    RHS_1=&\lim_{T\rightarrow\infty}\frac{1}{TN}\sum_{t=1}^{T}\sum_{i=1}^{N}\beta_{i}(1-p_i\mu^s_i)\mathbb{E}[\mathbb{I}_{L_i(t)=L}]\mathbb{E}\left[2h_{i}(t)-1| L_i(t)=L\right]
\end{aligned}
\end{equation}
\begin{equation}\label{RHS2}\begin{aligned}
    RHS_2=&\lim_{T\rightarrow\infty}\frac{1}{TN}\sum_{t=1}^{T}\sum_{i=1}^{N}\gamma_{i}(1-p_i\mu^s_i)\mathbb{E}[\mathbb{I}_{1<L_i(t)<L_i}]\mathbb{E}\left[2z_{i}(t)+2L_{i}(t)-1| L_i(t)>1\right]
\end{aligned}
\end{equation}
\begin{equation}\label{RHS3}
RHS_3=   \frac{NV}{2}
\end{equation}
Since $LHS_1>-\infty$, $LHS_2>0, RHS_1,RHS_2,RHS_3$ are finite, we establish that 
\begin{equation}
    LHS_3<\infty.
\end{equation}
Further simplifying~\eqref{LHS3} yields
\begin{equation}\label{throughput stable}
    \lim_{T\rightarrow\infty}\frac{1}{T}\sum_{t=1}^T \mathbb{E}[x_{i}^{+}(t)]<\infty,
\end{equation}
which is sufficient to establish that the throughput achieved by Max-Weight policy satisfies that $q_i^{MW}\geq \bar{q}_i, \forall i$.~\cite[Theorem 2.8]{neely2010stability}.

\section{Proof of Lemma~\ref{lem:Ci-ordering}}
\label{Proof of lemma Ci-ordering}
Fix a source $i$ and an update index $m$. Recall that $t'_i[m]$ denotes the
slot of the \emph{first} successful packet of update $m$. For notational
simplicity, we write $t_1 \triangleq t'_i[m]$. Thus,
($u_i(t_1)=1,c_i(t_1)=1,L_i(t_1)=L_i$), and let $t_i[m]$ denote the slot in
which the \emph{last} packet of the same update is delivered
($u_i(t_i[m])=1,c_i(t_i[m])=1,L_i(t_i[m])=1$). By the evolution of AoI~\eqref{hevolves}, we have
\[
h_i\bigl(t_i[m]\bigr)=h_i(t_1)+z_i\bigl(t_i[m]\bigr).
\]

Using \eqref{lyapunov drift upper bound3}, and the fact that at $t_1$ only the
term $\beta_i\mathbb{I}_{L_i(t)=L_i}\left(2h_i(t)-1\right)$ is active while at
$t_i[m]$ only the terms with $L_i(t)=1$ are active, we obtain
\begin{equation}\label{eq:Ci-last-first}
\begin{aligned}
&C_i\bigl(t_i[m]\bigr)-C_i(t_1)
= \beta_i\Bigl(h_i^2\bigl(t_i[m]\bigr)
      -2h_i\bigl(t_i[m]\bigr)z_i\bigl(t_i[m]\bigr)\Bigr)
   \\&\quad+\gamma_i\Bigl(\bigl(z_i\bigl(t_i[m]\bigr)+2\bigr)^2-(L_i+1)^2\Bigr) -\beta_i\bigl(2h_i(t_1)-1\bigr)
   +V\bigl(x_i^+\bigl(t_i[m]\bigr)-x_i^+(t_1)\bigr)\\[0.3em]
&\stackrel{(a)}{=}
   \beta_i\Bigl[(h_i(t_1)-1)^2-z_i^2\bigl(t_i[m]\bigr)\Bigr]+\gamma_i\Bigl(\bigl(z_i\bigl(t_i[m]\bigr)+2\bigr)^2-(L_i+1)^2\Bigr)\\
&\qquad
   +V\bigl(x_i^+\bigl(t_i[m]\bigr)-x_i^+(t_1)\bigr)\\[0.3em]
&\stackrel{(b)}{\ge}
   \beta_i\Bigl[(h_i(t_1)-1)^2-z_i^2\bigl(t_i[m]\bigr)\Bigr]+\beta_i\Bigl(\bigl(z_i\bigl(t_i[m]\bigr)+2\bigr)^2-(L_i+1)^2\Bigr)\\
&\qquad
   +V\bigl(x_i^+\bigl(t_i[m]\bigr)-x_i^+(t_1)\bigr)\\[0.3em]
&\stackrel{(c)}{\ge}
   \beta_i\Bigl[(h_i(t_1)-1)^2\!-\!(L_i-1)^2\Bigr]
   +V\bigl(x_i^+\bigl(t_i[m]\bigr)\!-x_i^+(t_1)\bigr)\\[0.3em]
&\stackrel{(d)}{\ge} 0.
\end{aligned}
\end{equation}
where $(a)$ follows from $h_i\bigl(t_i[m]\bigr)=h_i(t_1)+z_i\bigl(t_i[m]\bigr)$
and the identity $h^2-2hz=(h-1)^2-z^2$; $(b)$ from $\gamma_i\ge\beta_i$
together with $\bigl(z_i\bigl(t_i[m]\bigr)+2\bigr)^2-(L_i+1)^2\ge0$; $(c)$ from
$z_i\bigl(t_i[m]\bigr)\ge L_i-1$, which implies
$\bigl(z_i\bigl(t_i[m]\bigr)+2\bigr)^2-(L_i+1)^2
   \ge z_i^2\bigl(t_i[m]\bigr)-(L_i-1)^2$; and $(d)$ from the lower bound
$h_i(t_1)\ge L_i+1$ (except for the very first update, which can be handled
separately) and the boundedness of $x_i^+(t)$, so that for sufficiently small
$V>0$ the positive term
$\beta_i\bigl[(h_i(t_1)-1)^2-(L_i-1)^2\bigr]$ dominates
$V\bigl(x_i^+\bigl(t_i[m]\bigr)-x_i^+(t_1)\bigr)$.

\medskip
Next we compare the value of $C_i(t)$ at the last packet with the value at the
intermediate packets of the same update. Let $t_2$ denote the slot in which the
$(L_i-1)$-th packet of update $m$ is successfully delivered. Then
$u_i(t_2)=1$, $c_i(t_2)=1$, and $L_i(t_2)=2$. Since $z_i(t)+L_i(t)$ is
non-decreasing during the delivery of update $m$,
as shown in~\eqref{z+L evolution}, the $\gamma_i$–term in $C_i(t)$ with
$1<L_i(t)<L_i$ is maximized among the second to the $(L_i-1)$-th packets at
$t_2$.

We first consider the case $h_i^2(t_1)\ge z_i^2\bigl(t_i[m]\bigr)$. In this
case, starting from \eqref{lyapunov drift upper bound3} we obtain
\begin{equation}\label{eq:Ci-last-middle-case1-compact}
\begin{aligned}
C_i\bigl(t_i[m]\bigr)-C_i(t_2)&= \beta_i\Bigl(h_i^2\bigl(t_i[m]\bigr)
      -2h_i\bigl(t_i[m]\bigr)z_i\bigl(t_i[m]\bigr)\Bigr)+\gamma_i\Bigl(\bigl(z_i\bigl(t_i[m]\bigr)+2\bigr)^2-(L_i+1)^2\Bigr)\\
&\quad -\gamma_i\bigl(2z_i(t_2)+3\bigr)
   +V\bigl(x_i^+\bigl(t_i[m]\bigr)-x_i^+(t_2)\bigr)\\[0.3em]
&\stackrel{(a)}{=}
   \beta_i\Bigl(h_i^2(t_1)-z_i^2\bigl(t_i[m]\bigr)\Bigr)+\gamma_i\Bigl(\bigl(z_i\bigl(t_i[m]\bigr)+2\bigr)^2-(L_i+1)^2\Bigr)\\
&\quad
   -\gamma_i\bigl(2z_i(t_2)+3\bigr)
   +V\bigl(x_i^+\bigl(t_i[m]\bigr)-x_i^+(t_2)\bigr)\\[0.3em]
&\stackrel{(b)}{\ge}
   \gamma_i\Bigl(\bigl(z_i\bigl(t_i[m]\bigr)+2\bigr)^2-(L_i+1)^2\Bigr)-\gamma_i\bigl(2z_i(t_2)+3\bigr)
   +V\bigl(x_i^+\bigl(t_i[m]\bigr)-x_i^+(t_2)\bigr)\\[0.3em]
&\stackrel{(c)}{\ge}
   \gamma_i\Bigl(\bigl(z_i(t_2)+3\bigr)^2-(L_i+1)^2\Bigr)
   -\gamma_i\bigl(2z_i(t_2)+3\bigr)
   +V\bigl(x_i^+\bigl(t_i[m]\bigr)-x_i^+(t_2)\bigr)\\[0.3em]
&\stackrel{(d)}{=}
   \!\gamma_i\Bigl(\!\bigl(z_i(t_2)+2\bigr)^2\!\!\!-\!\!(L_i+1)^2\!\Bigr)
   \!+\!2\gamma_i
   \!+\!V\bigl(x_i^+\bigl(t_i[m]\bigr)\!\!-\!\!x_i^+(t_2)\bigr)\\[0.3em]
&\stackrel{(e)}{\ge} 0.
\end{aligned}
\end{equation}
Here $(a)$ uses $h_i\bigl(t_i[m]\bigr)=h_i(t_1)+z_i\bigl(t_i[m]\bigr)$; $(b)$
uses the assumption $h_i^2(t_1)\ge z_i^2\bigl(t_i[m]\bigr)$ to drop the
non-negative term
$\beta_i\bigl(h_i^2(t_1)-z_i^2\bigl(t_i[m]\bigr)\bigr)$; $(c)$ uses
$z_i\bigl(t_i[m]\bigr)\ge z_i(t_2)+1$ to replace $z_i\bigl(t_i[m]\bigr)$ by the
smaller value $z_i(t_2)+1$ in the positive quadratic term; $(d)$ follows from
the identity
$\bigl(z_i(t_2)+3\bigr)^2-(L_i+1)^2-(2z_i(t_2)+3)
 =\bigl(z_i(t_2)+2\bigr)^2-(L_i+1)^2+2$; and $(e)$ again uses the boundedness
of $x_i^+(t)$ and a sufficiently small $V>0$ so that the positive $2\gamma_i$
term dominates.

We now turn to the case $h_i^2(t_1)<z_i^2\bigl(t_i[m]\bigr)$. Since
$\gamma_i\ge\beta_i$, we can bound the first term in
\eqref{eq:Ci-last-middle-case1-compact} using $\gamma_i$ and obtain
\begin{equation}\label{eq:Ci-last-middle-case2-compact}
\begin{aligned}
C_i\bigl(t_i[m]\bigr)-C_i(t_2)
&\ge
   \gamma_i\Bigl(h_i^2(t_1)-z_i^2\bigl(t_i[m]\bigr)
   +\bigl(z_i\bigl(t_i[m]\bigr)+2\bigr)^2-(L_i+1)^2\Bigr)
   -\gamma_i\bigl(2z_i(t_2)+3\bigr)
   +V\bigl(x_i^+\bigl(t_i[m]\bigr)-x_i^+(t_2)\bigr)\\[0.3em]
&\stackrel{(a)}{\ge}
   \gamma_i\Bigl(h_i^2(t_1)-(z_i(t_2)+1)^2
   +\bigl(z_i(t_2)+3\bigr)^2-(L_i+1)^2\Bigr)
   -\gamma_i\bigl(2z_i(t_2)+3\bigr)
   +V\bigl(x_i^+\bigl(t_i[m]\bigr)-x_i^+(t_2)\bigr)\\[0.3em]
&\stackrel{(b)}{=}
   \!\!\gamma_i\Bigl(h_i^2(t_1)\!\!-\!(L_i+1)^2\!+\!2z_i(t_2)\!+\!5\Bigr)
   \!\!+\!\!V\bigl(x_i^+\bigl(t_i[m]\bigr)\!\!-\!\!x_i^+(t_2)\bigr)\\[0.3em]
&\stackrel{(c)}{\ge} 0.
\end{aligned}
\end{equation}
Here $(a)$ again uses $z_i\bigl(t_i[m]\bigr)\ge z_i(t_2)+1$; $(b)$ comes from
expanding the quadratic terms and using
$\bigl(z_i(t_2)+3\bigr)^2-(z_i(t_2)+1)^2=4z_i(t_2)+8$; and $(c)$ follows from
$h_i(t_1)\ge L_i+1$ and $z_i(t_2)\ge0$, which imply
$h_i^2(t_1)-(L_i+1)^2+2z_i(t_2)+5>0$, together with the boundedness of
$x_i^+(t)$ and a sufficiently small $V>0$.

Combining \eqref{eq:Ci-last-middle-case1-compact},
\eqref{eq:Ci-last-middle-case2-compact}, and \eqref{eq:Ci-last-first} completes the proof.

\section{Proof of Lemma~\ref{lem:C-last-avg}}
\label{Proof of lemma C-last-avg}
By Lemma~\ref{lem:Ci-ordering}, the value of $C_i(t)$ at the last successful
packet of update $m$ is no smaller than at any other successful packet of the
same update. Therefore, comparing $C_i(t_i[m])$ with the average of $C_i(t)$
over all packet delivery slots $t\in[t_i[m-1]+1,t_i[m]]$ yields
\begin{equation}\label{C ineq for each update}
\begin{aligned}
    &\frac{\sum_{t=t[m-1]+1}^{t[m]}C_{i}(t)\,
      \mathbb{I}_{L_{i}(t)=1}\mathbb{I}_{u_{i}(t)=1}\mathbb{I}_{c_{i}(t)=1}}
{\sum_{t=t[m-1]+1}^{t[m]}\mathbb{I}_{L_{i}(t)=1}\mathbb{I}_{u_{i}(t)=1}\mathbb{I}_{c_{i}(t)=1}}\ge
\frac{\sum_{t=t[m-1]+1}^{t[m]}C_{i}(t)\,
\mathbb{I}_{u_{i}(t)=1}\mathbb{I}_{c_{i}(t)=1}}
{\sum_{t=t[m-1]+1}^{t[m]}\mathbb{I}_{u_{i}(t)=1}\mathbb{I}_{c_{i}(t)=1}}.
\end{aligned}
\end{equation}

Next, we pass from the per-update inequality~\eqref{C ineq for each update} to
a time-averaged conditional expectation statement. Observe that each update $m$
comprises $L_i$ packet deliveries, and exactly one time slot corresponds to the
delivery of the last packet. Therefore, over the interval
$t\in[t[m-1]+1,\,t[m]]$ we have
\[
\sum_{t=t[m-1]+1}^{t[m]}
\mathbb{I}_{L_{i}(t)=1}\mathbb{I}_{u_{i}(t)=1}\mathbb{I}_{c_{i}(t)=1}
=1,\quad\forall m,\text{ and }
\sum_{t=t[m-1]+1}^{t[m]}
\mathbb{I}_{u_{i}(t)=1}\mathbb{I}_{c_{i}(t)=1}
=L_{i}, \quad\forall m.
\]
Using~\eqref{C ineq for each update}, this implies
\begin{equation}\label{C ineq for each update simplified}
\begin{aligned}
    \sum_{t=t[m-1]+1}^{t[m]}
    C_{i}(t)\mathbb{I}_{L_{i}(t)=1}\mathbb{I}_{u_{i}(t)=1}\mathbb{I}_{c_{i}(t)=1}
    \geq\frac{1}{L_{i}}\sum_{t=t[m-1]+1}^{t[m]}
    C_{i}(t)\mathbb{I}_{u_{i}(t)=1}\mathbb{I}_{c_{i}(t)=1} \; .
\end{aligned}
\end{equation}
Since~\eqref{C ineq for each update simplified} holds for every interval
$[\,t[m-1]+1,\,t[m]\,]$, summing it over all updates $m$ yields
\begin{equation}\label{C ineq for sum of updates}
\begin{aligned}
\sum_{m=1}^{D_{i}(T)}
\sum_{t=t[m-1]+1}^{t[m]}
C_{i}(t)\mathbb{I}_{L_{i}(t)=1}\mathbb{I}_{u_{i}(t)=1}\mathbb{I}_{c_{i}(t)=1}
\geq\frac{1}{L_{i}}\sum_{m=1}^{D_{i}(T)}\sum_{t=t[m-1]+1}^{t[m]}
C_{i}(t)\mathbb{I}_{u_{i}(t)=1}\mathbb{I}_{c_{i}(t)=1}.
\end{aligned}
\end{equation}
Letting the time horizon go to infinity, $T\rightarrow\infty$, we obtain
\begin{equation}\label{C ineq for infty T before expectation}
\begin{aligned}
\lim_{T\rightarrow\infty}\sum_{t=1}^{T}
C_{i}(t)\mathbb{I}_{L_{i}(t)=1}\mathbb{I}_{u_{i}(t)=1}\mathbb{I}_{c_{i}(t)=1}
\geq\lim_{T\rightarrow\infty}\frac{1}{L_{i}}\sum_{t=1}^{T}
C_{i}(t)\mathbb{I}_{u_{i}(t)=1}\mathbb{I}_{c_{i}(t)=1}.
\end{aligned}
\end{equation}

Assuming the ergodicity
of the network state $\mathbb{S}(t)$, we have
\[
\mathbb{E}\bigl[C_{i}(t)\mathbb{I}_{L_{i}(t)=1}\mathbb{I}_{u_{i}(t)=1}\mathbb{I}_{c_{i}(t)=1}\bigr]
\!\!\geq\!\!
\frac{1}{L_{i}}\mathbb{E}\bigl[C_{i}(t)\mathbb{I}_{u_{i}(t)=1}\mathbb{I}_{c_{i}(t)=1}\bigr].
\]

Moreover, since the instantaneous channel state $c_{i}(t)$ is independent of
$C_{i}(t)$, the update length $L_i(t)$, and the scheduling decision $u_i(t)$,
it follows that
\[
    \mathbb{E}\bigl[C_{i}(t)\mathbb{I}_{L_{i}(t)=1}\mathbb{I}_{u_{i}(t)=1}\bigr]
    \;\geq\;
    \frac{1}{L_{i}}\mathbb{E}\bigl[C_{i}(t)\mathbb{I}_{u_{i}(t)=1}\bigr].
\]
Using
$L_{i}\mathbb{E}[\mathbb{I}_{L_{i}(t)=1}\mathbb{I}_{u_{i}(t)=1}]
=\mathbb{E}[\mathbb{I}_{u_{i}(t)=1}]$, we finally obtain
\[
    \frac{\mathbb{E}[C_{i}(t)\mathbb{I}_{L_{i}(t)=1}\mathbb{I}_{u_{i}(t)=1}]}
         {\mathbb{E}[\mathbb{I}_{L_{i}(t)=1}\mathbb{I}_{u_{i}(t)=1}]}
    \;\geq\;
    \frac{\mathbb{E}[C_{i}(t)\mathbb{I}_{u_{i}(t)=1}]}
         {\mathbb{E}[\mathbb{I}_{u_{i}(t)=1}]},
\]
which is equivalent to~\eqref{C ineq for infty T}, i.e.,
\[
    \mathbb{E}\bigl[C_{i}(t)\mid L_{i}(t)=1,u_{i}(t)=1\bigr]
    \;\geq\;
    \mathbb{E}\bigl[C_{i}(t)\mid u_{i}(t)=1\bigr].
\]

\end{document}